\documentclass[twoside,11pt]{article} 
\usepackage{amssymb,latexsym,amsthm,amsopn,amstext,amscd,mathrsfs,amsmath}
\usepackage {graphicx}
\usepackage{enumitem}
\usepackage{amsfonts}
\usepackage{layout} 
\usepackage[all]{xy}

\theoremstyle{definition}

\theoremstyle{plain}
\newtheorem{tma}{Theorem}
 
\newtheorem{lma}{Lemma}
\newtheorem{pro}{Proposition} 

\newtheorem{cor}{Corollary}

\theoremstyle{remark}

\newcommand{\gf}{\varphi} 
\newcommand{\geps}{\varepsilon} 
\newcommand{\ran}{\mathrm{ran\,}}

\begin{document}
 
\title{Dirichlet to Neumann operator for abelian Yang-Mills gauge fields}
\author{Homero G. D\'\i az-Mar\'\i n\footnote{
{\tt hdiaz@umich.mx}
}
\\
Facultad de Ciencias F\'\i  sico-Matem\'aticas\\
Universidad Michoacana de San Nicol\'as de Hidalgo
\\
Ed. Alfa, Ciudad Universitaria, C.P. 58060, Morelia, M\'exico. 
}

\maketitle

\begin{abstract}
We consider the Dirichlet to Neumann operator for abelian Yang-Mills boundary conditions. The aim is constructing a complex structure for the symplectic space of boundary conditions of Euler-Lagrange solutions modulo gauge for space-time manifolds with smooth boundary. Thus we prepare a suitable scenario for geometric quantization within the reduced symplectic space of boundary conditions of abelian gauge fields .



{\bf Subject classification:} Yang-Mills theory, topological field theories, methods from differential geometry.

\end{abstract}

%%%%%%%%%%%%%%%%%%%%%%%%%%%%%%%%%%%%%%%%%%%%%%%%%%%%%%%%%%%%%%%%%%%%%%%%%%%

\section{Introduction}

Passing from Dirichlet to Neumann boundary conditions for a \emph{BVP}\footnote{Use this abbreviation for \emph{boundary value problem}}, on the boundary $\partial M$ of a smooth manifold $M$, is a classical problem for PDE. Some BVP may have an associated Dirichlet to Neumann (D-N) operator, so that for every solution its Neumann conditions can be recovered from its Dirichlet conditions and vice versa.  See for instance the D-N operator for case of the scalar Laplace equation in \cite{T} and  for $k-$forms in \cite{BS}.

{\bf General Boundary Formulation of Classical Fields.} 
In the context of Classical Field Theories, space-time {\em regions} are modeled as smooth oriented $n-$dimensional manifolds $M$ with boundary, we will consider compact regions for simplicity but regions may be general in a broader formalism. We consider BVP associated to the Euler-Lagrange equations of an action $S_M$. For linear and affine field theories, the precise description of the space of boundary conditions of solutions on every region $M$ allows us to adopt an axiomatic formulation, see \cite{O1, O, O2}. Thus for every region $M$ there exists an affine space of Euler-Lagrange solutions ${A}_M$, modeled by a linear space ${L}_M$. There exists a linear space ${L}_{\partial M}$ of boundary conditions, given by $1-$jets of solutions of the boundary on a cylinder $\partial M\times[0,\geps]\cong(\partial M)_\geps\subseteq M,$ $\geps>0$. Within this framework, dynamics is modeled as the Lagrangian embedding of the linear space $L_{\tilde{M}}\subseteq {L}_{\partial M}$ of boundary conditions of solutions in the interior of the region $M$, see \cite{CMR, CMR1, O1}. Here we exploit the existence of a presymplectic structure $\widetilde{\omega}_{\Sigma}$ in ${L}_{\Sigma}$, associated to every Lagrangian density, and every {\em hypersurface}, that is, on every $(n-1)-$dimensional smooth oriented closed submanifold, $\Sigma\subseteq M$, see for instance \cite{Zu}. In the absence of gauge symmetries, $\widetilde{\omega}_{\partial M}$ is non-degenerated. This is the classical part of a more general formalism referred as General Boundary Field Theories, or GBFT, see \cite{O3}.

{\bf Geometric Quantization.}
This setting is well suited for further implementation of geometric quantization. The clue to produce a polarization, and hence a geometric quantization procedure, in this case is precisely the D-N operator. It can be used to construct a complex structure for the space of boundary conditions decomposed as a direct sum of Dirichlet and Neumann conditions. As an example of this sort of the resulting QFT \cite{Seg:cftdef} suggests scalar quantum field theories in regions provided with a Riemannian metric, see also \cite{K}. This results rely on the dependence of the polarization on the first-order boundary data of the Riemannian metric of $M$ restricted to $\partial M$, see \cite{Segal:lectures}. This is related with conjectures and results claiming that the jet of the Riemannian metric (or the D-N operator of certain BVP) in the boundary $\partial M$, gives a complete characterization of the metric in the interior $M$, see \cite{SS} and references therein. The existence of a linear isomorphism between the kernel, $K$ and the range $R$ of the D-N operator  and a decomposition, ${L}_{\partial M}\simeq K\oplus R$, yields complex structure $J_{\partial M}$ in the symplectic space. A Hermitian structure arises from the tame complex structure with respect to the symplectic structure $\widetilde{\omega}_{\partial M}$. It leads to a Hilbert space needed for holomorphic prequantization, see \cite{O, O1, Wo}.  Hopefully this process will lead to an axiomatic framework for a quantum field theory called the General Boundary Quantum Field Theory GBQFT, as first proposed in \cite{O3}.

{\bf Gluing.}
The firsts axiomatic attempts for constructing quantum field theories within a categorical scenario were Topological Quantum Field Theories TQFT, see \cite{At,Seg:cftdef}. Within this setting hypersurfaces would be considered as objects while regions would be considered as morphisms, namely cobordisms that model time evolution. A TQFT can be then described as a "quantum functor" from the cobordisms category onto a suitable "Hilbert space Category". Despite some important examples of TQFTs, for most applications inconsistencies appear. These are usually avoided with the introduction of some technicalities or special cases. Similarly, the proposal of a "classical functor" from the category of morphisms onto a suitable "Symplectic Category" where objects would be symplectic vector spaces while morphisms would consist of Lagrangian correspondences has some prevailing technical difficulties. It is proposed that a corrective program sketched in \cite{We} may endure. GBFT (and its quantum counterpart GBQFT) drops the functor demand avoiding those issues. It just retains "gluing rules" from the axiomatics maintaining the  the predictive  tools of the theory. Instead of cobordisms we consider gluing a region $M$ along boundary components $\Sigma\cong\overline{\Sigma'},\, \Sigma,\Sigma'\subset \partial M$, to obtain a new region $M_1$. We then describe the relation of the spaces of solutions $L_{M_1}$ and boundary conditions $L_{\partial M_1}$ for the new region arising from those for the old region $M$. In the classical scenario this relations would model the space-time evolution, not just time evolution, of solutions.

{\bf General Boundary Formulation of Classical Gauge Fields.}  In general the presymplectic structure $\widetilde{\omega}_{\partial M}$ in ${L}_{\partial M}$, associated to the Lagrangian density is degenerate. In the linear and affine case, by taking the quotient by $\ker\widetilde{\omega}_{\partial M}$ gauge reduction can be achieved. This was explored  in a previous work \cite{DM} where we exposed an axiomatic approach in the GBFT formalism for abelian gauge fields applying it to the Yang-Mills  case, see also \cite{CMR, CMR1}.

{\bf Main results.} 
Along this work we restrict ourselves to the example of abelian Yang-Mills fields and follow the quantization program previously sketched for classical field theories. First we note that there exists a linear space  ${L}_{M,\partial M}\subseteq {L}_{\partial M}$ of topologically admissible solutions in the cylinder, depending on the topology and the metric of $M$ and of $\partial M$. Hence gauge reduction yields a symplectic space
\[
	\mathbf{L}_{M,\partial M}={L}_{M,\partial M}/{L}_{M,\partial M}^\bot.
\]
where ${L}_{M,\partial M}^\bot \subseteq {L}_{M,\partial M}$ stands for the symplectic orthogonal complement of ${L}_{M,\partial M}$ regarding the inherited presymplectic structure $\widetilde{\omega}_{\partial M}\vert_{{L}_{M,\partial M}}$, in ${L}_{M,\partial M}$. In the reduced space a symplectic structure $\omega_{\partial M}$ is induced by $\widetilde{\omega}_{\partial M}$. The space of boundary conditions modulo gauge of solutions, 
\[
	\mathbf{L}_{\tilde{M}}=\left({L}_{\tilde{ M}}\cap L_{M,\partial M}\right)/
						\left({L}_{\tilde{ M}}\cap		 {L}_{M,\partial M}^\bot\right)
						\subset \mathbf{L}_{M,\partial M}
\]
encodes the dynamics of the gauge fields. By considering the D-N operator, the complex structure is defined for gauge classes of the space of boundary conditions. We consider certain Yang-Mills BVP, see (\ref{eqn:BVP2-YM}), and construct the D-N operator, $\Lambda_{\tilde{M}}$, that transforms a Dirichlet condition, $\gf^D$, of a solution $\gf$, into its corresponding Neumann condition, $\gf^N=\Lambda_{\tilde{M}}(\gf^D)$. The existence of a linear isomorphism $\ker\Lambda_{\tilde{M}}\simeq \ran\Lambda_{\tilde{M}},$ and a decomposition, $\mathbf{L}_{M,\partial M}\simeq\ker\Lambda_{\tilde{M}}\oplus\ran\Lambda_{\tilde{M}}$, yields complex structure
\[
J_{\partial M}:\mathbf{L}_{M,\partial M}
\rightarrow
\mathbf{L}_{M,\partial M},\qquad J^2_{\partial M}=-\rm{id},
\]
see Theorem \ref{tma:J}. The aim is to apply geometric quantization tools for the holomorphic representation given by this complex structure. This will be done elsewhere.

{\bf Description of sections.} In section \ref{sec:1} we give a quick review of the classical abelian Yang-Mills theory emphasizing its GBFT formulation. A presymplectic structure is defined in the space of boundary conditions of solutions on the boundary. In section \ref{sec:2} we apply gauge reduction to obtain symplectic spaces of gauge fields, dynamics is then described in terms of a Lagrangian relation. This relation consists of the space of boundary conditions of gauge fields that are solutions. Most of the results on symplectic reduction and Lagrangian relations as described above are contained in Proposition \ref{cor:1}, Proposition \ref{lma:C} and Theorem \ref{tma:2}. Furthermore in Theorem \ref{cor:tma1} we show that the reduced space $\mathbf{L}_{M,\partial M}$ has finite codimension in the space $\mathsf{L}_{\partial M}$ consisting of gauge classes of boundary conditions of solutions on the cylinder $(\partial M)_{\varepsilon}$ of the boundary $\partial M$. In section \ref{sec:3} we propose the gluing laws that allow to reconstruct the space of solutions of space-time regions consisting of the gluing regions along boundary components. Gluing leads to a reduction of this codimension, $
\mathrm{codim}\,\mathbf{L}_{M_1,\partial M_1}$ is less or equal to  $\mathrm{codim}\,\mathbf{L}_{M,\partial M}.$ Intuitively this means that the topological homological complexity of $M$ expressed in $\partial M$ is increased by the gluing process. Section \ref{sec:J-structure} presents the main topic of this work, namely the definition of a complex structure for boundary conditions associated to the D-N operator. In \ref{sec:5} we define a Hermitian form on the space of boundary data which associates a complex Hilbert space to each boundary component.

%%%%%%%%%%%%%%%%%%%%%%%%%%%%%%%%%%%%%%%%%%%%%%%%%%%%%%%%%%%

%%%%%%%%%%%%%%%%%%%%%%%

\section{Classical abelian gauge Yang-Mills fields}\label{sec:1}

%%%%%%%%%%%%%%%%%%%%%%%

We consider an $n-$dimensional Riemannian smooth manifold $M$ with (smooth) boundary $\partial M$. We adopt it as a model for a space-time \emph{region}. We consider connections $\gf'$ in a principal fiber bundle with abelian fiber $U(1)$ on $M$ together with the Yang-Mills action
$$
	S_M(\gf')=\int_MF^{\gf'}\wedge \star F^{\gf'}
$$
where $F^{\gf'}$ denotes the curvature and $\star $ the Hodge star operator on $M$. We consider the space of Euler-Lagrange solutions $A_M$. By fixing a particular solution $\gf_0'\in A_M$, recall that there exists an identification of $\gf'\mapsto \gf:=\gf'-\gf_0'$, from the affine space $A_M$ consisting of connections to the corresponding linear space
$$
	L_M=\left\{\gf\in\Omega^1(M)\,:\,d^\star d\gf=0\right\}
$$
for a connection $\gf'$, its curvature is locally expressed as $d\gf$ and $d^\star$ denotes the codifferential with respect to the Hodge star operator. Gauge quotients in the abelian case are well defined. For the spaces obtained by gauge equivalence, we will apply the formalism for an affine field theory. Its axiomatic formalism is given in \cite{O1}.

%%%%%%%%%%%%%%%%%%%%%%%

\subsection{Gauge symmetries}
%%%%%%%%%%%%%%%%%%%%%%%

Consider the space of the Euler-Lagrange solutions modulo gauge
\[
	\mathsf{A}_M:=A_M/G_M^0
\]
where $G_{M}^0:=\{df\,:\, f\in\Omega^0(M)\}$ stands for the identity component of the gauge group $G_{M}$. The gauge action on the space of solutions is by translations. Hence the space $\mathsf{A}_M$ is affine with associated vector space $\mathsf{L}_M$. Once we have fixed a solution modulo gauge $\gf_0'\in\mathsf{A}_M$, there is an identification $[\gf']\mapsto [\gf]=[\gf']-[\gf_0']\in\mathsf{L}_M,\, \forall [\gf']\in\mathsf{A}_M$.

For a solution, $\gf\in{L}_M,$ take the boundary conditions
\begin{equation}\label{eqn:phi^D,phi^N}
	\gf^D:=i^*_{\partial M}\gf,\qquad \gf^N:=(-1)^n\star_{\partial M} i^*_{\partial M}(\star d\gf),\qquad \forall\gf\in L_M
\end{equation}
consisting of a Dirichlet condition, $\gf^D$, induced by the inclusion $i_{\partial M}$ of $\partial M$ into $M$, and a Neumann condition, $\gf^N$, for which the Hodge star operator, $\star_{\partial M}$, of the restricted Riemannian metric on $\partial M$ is considered. They define the space of boundary conditions of solutions $L_{\tilde{M}}$. The affine space $A_{\tilde{M}}$ is the space of boundary conditions of solutions in $A_M$. Consider the affine map and the corresponding linear map
\begin{gather}
	a_M(\gf')=\left((\gf')^D,(\gf')^N\right)\in A_{\tilde{ M}},\qquad \forall\gf'\in A_M\,\\
	r_M(\gf)=\left(\gf^D,\gf^N\right)\in L_{\tilde{ M}},\qquad \forall\gf\in L_M
\end{gather}
respectively. More generally,  let $r_{M}:\Omega^1(M)\rightarrow \left(\Omega^1(\partial M)\right)^{\oplus 2}$ be the projection to the boundary conditions,
\begin{equation}\label{eqn:p}
	r_{M}(\gf):=\left(\gf^D,\gf^N\right).
\end{equation}
For the \emph{closed} Riemannian manifold, $\partial M$, we have the Hodge decomposition
\begin{equation}\label{eqn:Hodge}
	\Omega^k(\partial M)=d\Omega^{k-1}(\partial M)\oplus\mathfrak{H}^k(\partial M)\oplus d^{\star_{\partial M}}\Omega^{k+1}(\partial M).
\end{equation}
Meanwhile, for the manifold with boundary $M$, recall that we have the Hodge-Morrey-Friedrichs (HMF) decomposition, see \cite{Sc}
\begin{equation}\label{eqn:HMF}
	\Omega^k(M)=d\Omega_D^{k-1}(M)\oplus\mathfrak{H}^k_N(M)\oplus
	\left(\mathfrak{H}^k(M)\cap d\Omega^{k-1}{M}\right)\oplus d^{\star}\Omega_N^{k+1}(M)
\end{equation}
where
\[
\begin{array}{rcl}
	\Omega_D^k(M)&:=&\left\{\alpha\,:\, \alpha\in\Omega^{k}(M)\,:\,\,i^*_{\partial M}\alpha=0\,\right\}	\\
	\Omega^k_N(M)&:=&\left\{\beta
			\,:	\, \beta\in\Omega^{k}(M)\,:\,i^*_{\partial M}\left(\star\beta\right)=0\,\right\}	\\
	\mathfrak{H}^k(M)&:=&\left\{\lambda\in\Omega^k(M)\,:\,d\lambda=0=d^{\star}\lambda\right\}\\
	\mathfrak{H}_D^k(M)&:=&\mathfrak{H}^k(M)\cap \Omega^k_D(M)\\
	\mathfrak{H}_N^k(M)&:=&\mathfrak{H}^k(M)\cap \Omega^k_N(M).\\
\end{array}
\]
Recall also that the differential $d$ acts on the chain complex $ \Omega_D^k(M)$, meanwhile the codifferential $d^{\star}$ acts on the complex $\Omega_N^k(M)$. The space of harmonic forms, $\mathfrak{H}^k(M)$, is infinite dimensional. Nevertheless its boundary conditioned subspaces, $ \mathfrak{H}^k_N(M),\mathfrak{H}^k_D(M)\subset \mathfrak{H}^k(M)$ are finite dimensional.

We define the \emph{axial gauge fixing space in the bulk} as
\begin{equation}\label{eqn:gaugefixingM}
	\Phi_{A_M}:=L_M\cap\left(\mathfrak{H}^1_N(M)\oplus d^{\star}\Omega_N^{2}(M)\right)\simeq \mathsf{L}_M
\end{equation}
Notice that the isomorphism follows from the fact that the $G_M^0-$orbits are transverse to $\Phi_{A_M}$, see \cite{T} Vol. II Section 9 on divergence-free vector fields.

\begin{pro}\label{lma:(d)}
For every solution, $\gf\in L_M$, its Neumann condition, $\gf^N\in\Omega^1(M),$ is coclosed, i.e.
\[
	d^{\star_{\partial M}}\gf^N=0\qquad\forall \phi=\left(\gf^D,\gf^N\right)\in L_{\tilde{M}}.
\]
%Hence\[	r_M(\Phi_{A_M})\simeq L_{\tilde{M}}.\]
\end{pro}

\begin{proof}
For every $\gf\in\Omega^1(M)$ there exists an extension $\widehat{\gf^N}\in\Omega^1(M)$ of $\gf^N\in\Omega^1(\partial M)$, such that $\gf^N=\left(\widehat{\gf^N}\right)^N$, according to Lemma \ref{lma:(c)} below. This means that $i^*_{\partial M}(\star d\gf)=i^*_{\partial M}\left(\star d\widehat{\gf^N}\right)$. This last assertion implies that
\[
		\star_{\partial M}i^*_{\partial M}\left(d\star d{\gf}\right)=
			\star_{\partial M} d i^*_{\partial M}\left(\star d{\gf}\right)=
		\]
		\[
		=\star_{\partial M} d i^*_{\partial M}\left(\star d\widehat{\gf^N}\right)=
		\star_{\partial M}i^*_{\partial M}\left(d\left(\star d\widehat{\gf^N}\right)\right).
\]
Thus by the second part of Lemma \ref{lma:(c)} below:
\[
	(-1)^{n}\cdot \star_{\partial M} d^{\star_{\partial M}}\gf^N=	\star_{\partial M}i^*_{\partial M}\left(d\star d{\gf}\right)=0.
\]
Last line follows from the fact that $\gf\in L_M$. Therefore $\gf^N\in\ker d^{\star_{\partial M}}.$
\end{proof}

Before we prove Lemma \ref{lma:(c)} let us fix some notation. Take the components of the boundary of a region $M$
\begin{equation}\label{eqn:dMdecomposition}
	\partial M=\Sigma^1\sqcup\dots\sqcup\Sigma^m,
\end{equation}
If we consider the \emph{tubular neighborhood}, $\partial M_\geps\cong \partial M\times[0,\geps],\,\geps >0$, then its boundary decomposes as
$$
\partial (\partial M_\geps)=
\partial M\sqcup{\partial M'}
\cong\partial M\sqcup\overline{\partial M}
$$
where ${\partial M'}$ is homeomorphic to $\partial M\times\{\geps\}$ with the orientation induced by $M$ and where $\overline{\partial M}$ denotes the same manifold $\partial M$ with inverted orientation.

Without reference to a region $M$, we define a \emph{hypersurface} $\Sigma$ as an $(n-1)-$dimensional {oriented connected closed} smooth manifold. It is provided with a tubular neighborhood, $\Sigma_\geps$, diffeomorphic to the cylinder $\Sigma\times [0,\geps],\, \geps >0$. When we refer to $\Sigma$ as a boundary component of a region we consider $\Sigma_\geps\subset M$. The boundary of $\Sigma_\geps$, consists of two diffeomorphic components, 
\[
	\partial\left( \Sigma_\geps\right)=\Sigma\sqcup\Sigma'.
\]
It also has a Riemannian structure on $\Sigma_\geps$. We suppose that there is a diffeomorphism 
\begin{equation}\label{eqn:X}
	X:\Sigma\times[0,\geps]\rightarrow\Sigma_\geps	
\end{equation}
where $\geps>0$ is small enough so that for every initial condition, $s\in \Sigma$, the curve, $\tau\mapsto X(s,\tau),\,0\leq \tau\leq\geps$, is a geodesic normal to $\Sigma$. Namely, consider $X$ as the exponential map, so that these geodesics foliate $\Sigma_\geps$. Define a diffeomorphism
\[
	X^{\geps}_\Sigma(\cdot):=X(\cdot,\geps):\Sigma\rightarrow\overline{\Sigma'}.
\]
where $\overline{\Sigma'}$ means that we consider $\Sigma'$ as a manifold with orientation inverted with respect to that induced by $\Sigma_\geps$.

\begin{lma}\label{lma:(c)}
For every $\phi\in\Omega^1(\partial M)$ there exists an extension $\widehat{\phi}\in\Omega^1(M)$ such that
\[
	\star_{\partial M}i^*_{\partial M}\left(\star d\widehat{\phi}\right)=(-1)^n\cdot \phi,\qquad
	\star_{\partial M}i^*_{\partial M}\left(d\left(\star d\widehat{\phi}\right)\right)=(-1)^{n-1}\cdot \star_{\partial M} d^{\star_{\partial M}}\phi.
\]
\end{lma}

\begin{proof}
Take $\phi\in\Omega^k(\partial M)$, and $\psi_\geps:M\rightarrow[0,1]$ defined as
\begin{equation}\label{eqn:psieps}
	\psi^{-1}_\geps(1)=\partial M, \,d\psi_\geps\mid_{\partial M}=0,\, \psi_{\geps}^{-1}(0)= (\partial M)'=X_{\partial M}^{\geps}(\partial M).
\end{equation}
Define
\[
	\widehat{\phi}:= \psi_{\geps}\cdot \tau \cdot \overline{\phi}(x,\tau)\in\Omega^1(M)
\]
where $\overline{\phi}\in\Omega^1(\partial M_\geps)$ is defined as
\begin{equation}\label{eqn:overphi}
	\overline{\phi}(x,\tau):= \left(X_{\partial M}^{-\tau}\right)^*\phi(x),
		\qquad \forall x\in\partial M,\qquad 0\leq \tau \leq \geps.
\end{equation}
here $(X^{-\tau}_{\partial M}):=(X^\tau_{\partial M})^{-1}$. In local expressions,
\[
\widehat{\phi}\mid_{\partial M_\geps}=\psi_\geps\cdot \tau \cdot \sum_{j=1}^{n-1}\phi_j(x)dx^j
\]
hence
\[
	\star_{\partial M}i^*_{\partial M}\left(\star d\widehat{\phi}\right)=
	\star_{\partial M}\sum_{j=1}^{n-1}\phi_j(x)i^*_{\partial M}\left(\star d\tau\wedge dx^j\right)=
\]
\[
=	\star_{\partial M}\sum_{i=1}^{n-1}\phi_j(x)\left(\star_{\partial M} dx^j\right)
		=	 (\star_{\partial M}\star_{\partial M})\cdot \phi=(-1)^{(n-2)}\cdot\phi=(-1)^{n}\cdot \phi,
\]
here we use the relation $i^*_{\partial M}(\star d\tau\wedge dx^j)=\star_{\partial M} dx^j$. This follows from local considerations as follows:
\[
	i^*_{\partial M}(\star d\tau\wedge dx^j)=\sqrt{\left|\det(h_{jk})\right|}\cdot 
	h^{1,j}\cdots h^{n-1,j}(-1)^j\cdot
\]
\[
	\cdot \left(dx^1\wedge\dots\wedge d\check{x}^j\wedge\dots\wedge dx^{n-1}\right)
	=
\]
\[
=
	\sqrt{\left|\det\left(\overline{h}_{jk}\right)\right|}\cdot
	\overline{h}^{1,j}\cdots \overline{h}^{n-1,j}(-1)^j\left(dx^1\wedge\dots\wedge d\check{x}^j\wedge\dots\wedge dx^{n-1}\right)
	=
\]
\[
	=\star_{\partial M}dx^j
\]
where $h_{jk}$ denotes the Riemannian metric in $\partial M_\geps$, while $\overline{h}_{jk}$ is the metric induced in $\partial M$. By the orthogonality condition for geodesics, $h^{j,n}=0,h^{nn}=1$. 

Furthermore
\[
	i^*_{\partial M}\left(d\star d\widehat{\phi}\right)=
	di^*_{\partial M}\left(\star d\widehat{\phi}\right)=
\]
\[
	d\left(\sum_{j=1}^{n-1}\phi_j(x)i^*_{\partial M}\left(\star d\tau\wedge dx^j\right)\right)=
	(-1)^{n-2}\star_{\partial M}\star_{\partial M}\sum_{j=1}^{n-1}d\left(\phi_j(x)\cdot \star_{\partial M} dx^j\right)=
\]
\[
	=(-1)^{n-1}\star_{\partial M} d^{\star_{\partial M}}\phi.
\]
Recall that $\star_{\partial M}\star_{\partial M}=(-1)^{n-2}$ and $d^{\star_{\partial M}}=(-1)^{(n-1)2+1}\star_{\partial M}\circ d\circ\star_{\partial M}=(-1)\cdot \star_{\partial M}\circ d\circ\star_{\partial M}$.
\end{proof}

%%%%%%%%%%%%%%%%%%%%%

\subsection{Boundary conditions on hypersurfaces}

%%%%%%%%%%%%%%%%%%%%%%%%

For a region $M$, let us consider the {BVP}
\begin{equation}\label{eqn:BVP1}
\left\{\begin{array}{ll}
\Delta\gf=0, & 
\\
i_{\partial M}^*\gf=\phi^D, &i_{\partial M}^*(d^{\star}\gf)=0
\end{array}.\right.
\end{equation}
According to \cite{Sc} for every $\phi^D\in\Omega^k(\partial M)$ there exists a solution $\gf\in\Omega^k(M)$ of (\ref{eqn:BVP1}).  The solution $\gf$ is unique up to $\lambda\in\mathfrak{H}^k_D(M)$. Furthermore, $\gf\in\mathfrak{H}^k(M)$, see Proposition 3.4.5 in \cite{Sc}. Recall that in the case $\partial M\neq\emptyset$ the space $\mathfrak{H}^k(M)$ is infinite dimensional and is different from the space of harmonic forms i.e. solutions of the Laplace equation, $\Delta\gf=0$. Furthermore, the BVP (\ref{eqn:BVP1}) is equivalent to the following BVP, see \cite{BS}  
\begin{equation}\label{eqn:BVP2}
\left\{\begin{array}{ll}
\Delta\gf=0, & d^{\star}\gf=0
\\
i_{\partial M}^*\gf=\phi^D, &
\end{array}\right.
\end{equation}
therefore $dd^{\star}\gf=0$ and $\Delta\gf=0,$ thus $d^{\star}d\gf=0$. Hence $\gf\in L_M$ and $\phi^D\in i^*_{\partial M} L_M$. Thus every solution to (\ref{eqn:BVP2}) at the same time solves to the following Yang-Mills BVP
\begin{equation}\label{eqn:BVP3-YM}
\left\{\begin{array}{ll}
	d^\star d\gf=0, & d^{\star}\gf=0
	\\
	i_{\partial M}^*\gf=\phi^D, & 
\end{array}\right. .
\end{equation}
Moreover every $\gf\in \Phi_{A_M}$ is a solution of this BVP.

When $M=\Sigma_\geps$ similar arguments can be adapted in order to prove the following assertion.

\begin{lma}\label{lma:Dirichlet}
Let ${L}_\Sigma={L}_\Sigma^D\oplus{L}_\Sigma^N$ be the linear space of boundary conditions $\phi=(\gf^D,\gf^N)\in L_\Sigma$ of solutions, $\gf\in L_{\Sigma_\geps
},$ in the cylinder $\Sigma_\geps$. Then the Dirichlet condition $\gf^D\in\Omega^1(\Sigma)$ can be any $1-$form in $\Sigma$,
\[
 	{L}_\Sigma^D= \Omega^1(\Sigma)
\]
\end{lma}

\begin{proof}

We consider solutions in $\Sigma_\geps$, whose boundary conditions are defined \emph{only in the bottom boundary component}, $\Sigma\cong\Sigma\times\{0\}$, of $\partial \Sigma_\geps$. Thus we define $A_\Sigma$ as the affine space of pairs $\phi'=\left((\phi')^D,(\phi')^N\right)$ as we did in (\ref{eqn:phi^D,phi^N}). Denote its corresponding linear space as $L_\Sigma$. Here we consider the inclusion of one component $i_\Sigma:\Sigma\rightarrow \Sigma_\geps$ instead of the inclusion of the whole boundary $i_{\partial \Sigma_\geps}:\partial\Sigma_\geps \rightarrow \Sigma_\geps$.

Recall that the following BVP has a solution, see \cite{Sc} Lemma 3.4.7, and \cite{BS} Lemma 3,
\begin{equation}\label{eqn:BVP-Sigma}
\left\{\begin{array}{ll}
	d^\star d\gf=0, & d^{\star}\gf =0
	\\
	i_{\Sigma}^*\gf=\phi^D,& i_{\Sigma'}^*\gf=\left(X_\Sigma^{-\geps}\right)^*\phi^D,
	\\
\end{array}\right.
\end{equation}
for every $\phi^D\in\Omega^k(\Sigma),$ where $i_\Sigma:\Sigma\rightarrow\Sigma_\geps$, $i_{{\Sigma'}}:\Sigma'\rightarrow\Sigma_\geps$ denote inclusions and $X^{-\geps}_\Sigma=(X^\geps_\Sigma)^{-1}$.  Notice that the boundary condition, $\phi^D$, is prescribed \emph{just in one component,} $\Sigma \subsetneq\partial \Sigma_\geps$. This proves that $L_\Sigma^D=\Omega^1(\Sigma)$. \end{proof}

\begin{lma}\label{lma:Neumman}
Let ${L}_\Sigma={L}_\Sigma^D\oplus{L}_\Sigma^N$ be the linear space of boundary conditions, $\phi=(\gf^D,\gf^N)\in L_\Sigma,$ of solutions, $\gf\in L_{\Sigma_\geps
},$ in the cylinder $\Sigma_\geps$. Then the Neumman condition $\gf^N$, is coclosed, $d^{\star_\Sigma}\gf^N=0$. Furthermore

\[\qquad {L}_\Sigma^N= \ker d^{\star_\Sigma}.\]

\end{lma}

We have already shown in Proposition \ref{lma:(d)} that $L_\Sigma^N\subseteq\ker d^{\star_\Sigma}$, here $\star_\Sigma$ stands for the Hodge star operator in $\Sigma$. The complete proof will follow from Lemma \ref{lma:6} and the isomorphism described in Lemma \ref{lma:previous} below, where $M=\Sigma_\geps$.

There is a presymplectic structure in $\left(\Omega^1(\Sigma)\right)^{\oplus 2}$ inducing a presymplectic structure in ${L}_\Sigma\subset \left(\Omega^1(\Sigma)\right)^{\oplus 2}$ given by
\begin{equation}\label{eqn:symplectic-structure}
	\begin{array}{c}
		\widetilde{\omega}_{\Sigma}\left(\phi_1,\phi_2\right)=
		\frac{1}{2}\left(\left[\phi_1,\phi_2\right]_{\Sigma}-\left[\phi_2,\phi_1\right]_{\Sigma}\right)
	\end{array}
\end{equation}
for every $\phi_i=\left(\phi^D_i,\phi^N_i\right)\in \left(\Omega^1(\Sigma)\right)^{\oplus2}$, see \cite{CW}, \cite{Zu}. Here we use the bilinear map:
\begin{equation}\label{eqn:[.,.]}
		\left[\phi_1,\phi_2\right]_{\Sigma}:=
		\int_{\Sigma}\phi^D_1\wedge\star_{\Sigma}\phi^N_2.
\end{equation}
In fact the $1-$form
\[
	\theta_{\Sigma}(\eta,\phi)=
		\int_{\Sigma}(\eta-\eta_0)^D\wedge\star_{\Sigma} \phi^N,\qquad
		\forall\eta\in {A}_{\Sigma},\, \forall \phi\in{L}_{\Sigma}
\]
is a  symplectic potential for the translation invariant presymplectic structure in the affine space, also denoted as $\widetilde{\omega}_{\Sigma}$. It also satisfies the translation invariance condition
\begin{equation}\label{eqn:theta}
	\left[\phi_1,\phi_2\right]_{\Sigma}+\theta_{\Sigma}\left(\eta,\phi_2\right)=
	\theta_{\Sigma}\left(\phi_1+\eta,\phi_2\right),\, \forall \eta\in {A}_{\Sigma},\,\forall\phi_1,\phi_2\in{L}_{\Sigma}.
\end{equation}

\begin{lma}
The degeneracy space of the bilinear form, $\widetilde{\omega}_\Sigma$ can be described as
\[
	\ker\widetilde{\omega}_\Sigma=		\left\{
			\left(df^D,dg^N	\right)\,:\, \left(f^D,g^N\right)\in \left(\Omega^0(\Sigma)\right)^{\oplus 2}
		\right\}\subset \left(\Omega^1(\Sigma)\right)^{\oplus 2}.
\]
\end{lma}

\begin{proof}
Take $\gf_0\in\ker \widetilde{\omega}_\Sigma$, then
\[
	\int_\Sigma\gf_0^D\wedge\star_\Sigma\gf^N-\gf^D\wedge\star_\Sigma\gf^N_0=0,\qquad
		\forall \gf\in{L}_\Sigma.
\]
According to Lemma \ref{lma:6} for every Dirichlet coclosed condition $\gf^D,\, d^{\star_\Sigma}\gf^D=0$, there exists a solution $\gf\in L_{\Sigma_\geps}$ such that $\gf^N=0$. It follows that
\[
	\int_\Sigma\gf^D\wedge\star_\Sigma\gf^N_0=0,\qquad
		\forall \gf^D\in\ \ker d^{\star_\Sigma}
\]
Hence $\gf^N_0$ is orthogonal to $\ker d^{\star_\Sigma}\subset \Omega^0(\Sigma)$. By Hodge decomposition applied to $\Omega^1(\Sigma)$, $\gf^N_0$ is exact.

Similarly for every coclosed Neumann condition $\gf^N\in\ker d^{\star_\Sigma}$ there exists a solution $\gf\in L_{\Sigma_\geps}$ such that $\gf^D=0$ and
\[	
	\int_\Sigma\gf_0^D\wedge\star_\Sigma\gf^N=0,\qquad
	\forall \gf^N\in\ \ker d^{\star_\Sigma}.
\]
Hence $\gf_0^D$ is exact. \end{proof}

When we consider the gauge group $G_{\Sigma_\geps}^0$ acting linearly, $\gf\mapsto \gf+ df$, $f\in\Omega^0(\Sigma_\geps),$ on equivalence classes $[\gf]_\Sigma\subset L_{\Sigma_\geps}$. Then
\begin{equation}
	\gf_1\sim_\Sigma\gf_2\text{ iff }\left(\gf_1^D,\gf_1^N\right)=\left(\gf_2^D,\gf_2^N\right)\in {L}_\Sigma,\qquad \forall \gf_i\in L_{\Sigma_\geps}.
\end{equation}
Take the quotient by the stabilizer of the action to obtain the \emph{gauge group for boundary conditions}, $G_\Sigma^0$ that \emph{does not depend on} $\geps>0$.

%\begin{lma}\label{lma:divergenceboundary}
%Take $\gf\in\Omega^1(M)$ if the Neumann condition vanishes, i.e., $i_{\partial M}^*(\star d\gf)=$, then the divergence on the boundary equals the divergence on the boundary  of the restriction
%\[
%	d^\star\gf\mid_{\partial M}=d^{\star_{\partial M}}\left(i^*_{\partial M}\gf\right)=i^*_{\partial M}\left(d^\star \gf\right).
%\]
%\end{lma}

%For the case of Yang-Mills boundary conditions: $i_{\partial M}^*\left(\star d(\star d\gf)\right)=0,$ hence
%\[
%	d^{\star_{\partial M}}i^*_{\partial M}\left(\star d\gf\right)=i_{\partial M}\left(d^\star\star d\gf\right)=0
%\]
%therefore the action the Neumann condition is coclosed: $\gf^N\in\ker d^{\star_{\partial M}}$, therefore 

The $G_\Sigma^0-$action on the Neumann boundary condition is trivial. On the other hand, the action on the Dirichlet condition can be given explicitly as
\[
		\left(\phi^D,\phi^N\right)\mapsto 
		\left((\phi+df)^D,(\phi+df)^N\right)=\left(\phi^D+d(f^D),\phi^N\right)
\]
where $\phi=\left(\phi^D,\phi^N\right)\in {L}_\Sigma,\, d(f^D)\in d\Omega^0(\Sigma).$ Hence we have the following assertion.

\begin{pro}
There is an isomorphism
\[
	{G}_\Sigma^0\simeq \ker\widetilde{\omega}_\Sigma	\cap L_\Sigma.
\]
where the orbit of 0 under the (free) $G^0_\Sigma-$action is identified with the gauge group $G^0_\Sigma$.
\end{pro}
Taking  \emph{the axial gauge} as the gauge fixing space of boundary conditions on the hypersurface, $\Sigma$,
\begin{equation}\label{eqn:gaugefixingSigma}
	\Phi_{A_\Sigma}:=\left\{\left(\phi^D,\phi^N\right)\in L_\Sigma\,:
		\, d^{\star_\Sigma}\phi^D=0=d^{\star_\Sigma}\phi^N\right\}\subset L_\Sigma
\end{equation}
we have that $\Phi_{A_\Sigma}\subset L_\Sigma$ is a linear subspace transverse to the $G^0_\Sigma-$orbits. For the cylinder $\Sigma_\geps$, \emph{every} Dirichlet and Neumann boundary conditions modulo gauge can be described by coclosed forms on the boundary, i.e.
\begin{equation}\label{eqn:L}
	\mathsf{L}_{\Sigma}:=L_\Sigma/G_{\Sigma}^0\simeq(\ker d^{\star_\Sigma})^{\oplus2}\simeq T(\ker d^{\star_{\Sigma}})\simeq \Phi_{A_\Sigma}.
\end{equation}
This will be proved in Lemma \ref{lma:6}.

The linear space $L_{\Sigma}$ with its presymplectic structure $\widetilde{\omega}_\Sigma$, yields a symplectic structure in the reduced space, $\mathsf{L}_{\Sigma}$. We call it $\omega_{\Sigma}$.

%%%%%%%%%%%%%%%%%%%%

\subsection{Regions and hypersurfaces}

%%%%%%%%

Take the components of the boundary $\partial M$ of a region $M$ as hypersurfaces as in (\ref{eqn:dMdecomposition}). We denote the affine  space of boundary conditions and its linear counterpart as
\begin{gather}\label{eqn:sumadirecta}
	A_{\partial M}=A_{\Sigma^1}\times\dots\times A_{\Sigma^m},\qquad
	L_{\partial M}={L}_{\Sigma^1}\oplus\dots\oplus{L}_{\Sigma^m}.
\end{gather}
We consider the gauge action $G_{\partial M_\geps}^0$ onto equivalence classes of solutions, $[\gf]\subset A_{\partial M_{\geps}}$, where $\gf_1\sim\gf_2$ iff $(\gf_1^D,\gf_1^N)=(\gf_2^D,\gf_2^N)\in A_{\partial M}$. By the inclusion $\partial M_\geps \subset M $ there is a compatibility of gauge actions in the bulk and in the boundary i.e. morphisms $G_{M}\rightarrow G_{\partial M_\geps}^0$. Therefore there is a well defined gauge group morphisms, $G_{M}\rightarrow G_{\partial M}^0$, explicitly $df\mapsto \left(d(i^*_{\partial M}f),0\right)$. Notice that 
\[
	d\Omega^0_D(M)\simeq\ker\left(G_M^0\rightarrow G_{\partial M}^0\right)\subset\ker r_M.
\]
There is also compatibility of gauge actions whose quotients are
\begin{gather}
	\mathsf{A}_{\partial M} :=A_{\partial M}/G_{\partial M}^0,\qquad \mathsf{L}_{\partial M} :=L_{\partial M}/G_{\partial M}^0.
\end{gather}

The affine and linear maps from the space of solutions to the corresponding boundary conditions $a_M:A_M\rightarrow A_{\partial M}$ and $r_M:L_M\rightarrow L_{\partial M}$ are compatible with the corresponding gauge group actions in the bulk and in the boundary respectively, see \cite{DM} axiom (A8). There is also a section $G_{\partial M}^0\rightarrow G_M^0$, see (\ref{eqn:tilde f}). Hence there are maps from the space of gauge fields in the interior to the space of Dirichlet-Neumann boundary conditions modulo gauge:
\begin{equation}\label{eqn:A_M,r_M}
		\mathsf{a}_M:\mathsf{A}_M\rightarrow\mathsf{A}_{\partial M},\qquad
		\mathsf{r}_M:\mathsf{L}_M\rightarrow\mathsf{L}_{\partial M}.
\end{equation}
Notice that $\mathsf{r}_M(\mathsf{L}_M)\simeq r_M(\Phi_{A_M})$.
Take the gauge fixing for hypersurfaces (\ref{eqn:gaugefixingSigma}) and the Hodge decomposition of $\ker d^{\star_{\partial M}}$, then  the \emph{axial gauge fixing space on the boundary} is
\begin{equation}
	\Phi_{A_{\partial M}}:=\left[\mathfrak{H}^1(\partial M)\oplus d^{\star_M}\Omega^{2}(\partial M)\right]^{\oplus2}.
\end{equation}
Recall also the linear isomorphism $\Phi_{A_{\partial M}}\simeq \mathsf{L}_{\partial M}$.

%%%%%%%%%%%%%%%%%%%%%

\section{Gauge reduction}\label{sec:2}

%%%%%%%%%%%%%%%%%%%%%%%%

Now we proceed to describe the symplectic reduction for the space $L_{\partial M}$ of boundary conditions of solutions in the cylinder, $\partial M_\geps$, in more detail. Consider the direct sum decomposition (\ref{eqn:sumadirecta}). We refer to the presymplectic structure $\widetilde{\omega}_{\partial M}=\widetilde{\omega}_{\Sigma^1}\oplus\dots\oplus\widetilde{\omega}_{\Sigma^m}$.

Define the space of \emph{topologically admissible boundary conditions} as
\begin{equation}\label{eqn:def C}
	{L}_{M,\partial M}:=\left\{r_M(\gf)\,:\, \gf\in\mathfrak{H}_N^1(M)\oplus d^{\star}\Omega^2_N(M),\, d^{\star_{\partial M}}\gf^N=0\right\}.
\end{equation}
The space $L_{\partial M}$ depends just on the germ of the Riemannian metric of the cylinder $\partial M_\geps$ restricted to $\partial M$ and does \emph{not} depend on the topology of $M$. Nevertheless the subspace ${L}_{M,\partial M}$ depends on the metric on the boundary but also depends on the relative topology of $M$ and $\partial M$.
Notice also that $r_M(\Phi_{A_M})\simeq L_{\tilde{M}}\cap {L}_{M,\partial M}.$

Known results, see for instance (2.1) in \cite{BS}, also \cite{Sc}, elucidate some topological issues in terms of the De Rham cohomology of $M$ related to the coisotropic space ${L}_{M,\partial M}$.

%%%%

Notice that the presymplectic space ${L}_{M,\partial M}$ contains its linear symplectic orthogonal complement ${L}_{M,\partial M}^\bot\subset {L}_{M,\partial M},$  regarding the presymplectic structure $\widetilde{\omega}_{\partial M}\mid_{L_{M,\partial M}}$. Hence we can consider it as a coisotropic space. The gauge quotient  $\mathsf{L}_{\partial M }$ is symplectic, when we consider the linear action of $G_{\partial M}^0=\ker\widetilde{\omega}_{\partial M}\cap L_{\partial M}$ onto $L_{\partial M}$. This can be summarized in the following assertion.

\begin{pro}\label{cor:1} The following are true:
\begin{enumerate}

\item There is an isomorphism 
\[
	{L}_{M,\partial M}^\bot=
		{L}_{M,\partial M}\cap \ker\widetilde{\omega}_{\partial M}\simeq {L}_{M,\partial M}\cap G^0_{\partial M}.
\]
\item The quotient space 
\begin{equation}\label{eqn:Lsymplectic}
 	\mathbf{L}_{M,\partial M}:={L}_{M,\partial M}/{L}_{M,\partial M}^\bot
\end{equation}
is a symplectic linear space.

\item There is an inclusion of spaces $L_{\tilde{M}}\cap {L}_{M,\partial M}^\bot\subset L_{\tilde{M}}\cap {L}_{M,\partial M}$.
\end{enumerate}
\end{pro}

The space of boundary conditions of solutions $A_{\tilde{M}}$, has corresponding linear space $L_{\tilde{M}}=r_M(L_M)$. The aim is to show that we have a Lagrangian subspace, $\mathbf{L}_{\tilde{M}}\subset \mathbf{L}_{M,\partial M}; $ where
\begin{equation}\label{eqn:Lagrangian}
	\mathbf{L}_{\tilde{M}}:=\left(L_{\tilde{M}}\cap {L}_{M,\partial M}\right)/\left(L_{\tilde{M}}\cap {L}^\bot_{M,\partial M}\right).
\end{equation}
This is consistent with the general setting of describing dynamics as Lagrangian relations in linear symplectic spaces, see \cite{We}.

%%%%

%\item The pullback $i^*_{\partial M}:H^1_{dR}(M)\rightarrow H^1_{dR}(\partial M)$ makes following diagram commutative
%	\[\xymatrix{
%		\ker\Lambda_{\tilde{M}}
%				\ar@{<->}[r]	\ar@{^{(}->}[d]	\ar@/_1cm/[ddd]_{d}	
%			&
%		d^{\star_{\partial M}}\Omega^2(\partial M)+i^*_{\partial M}H_{dR}^1(M)
%				\ar@{^{(}->}[d]	
%			&
%		H_{dR}^1(M)
%				\ar[l]_{i^*_{\partial M}}	\ar[d]_{i^*_{\partial M}}	\ar@/^1.5cm/[ddd]^{d}
%		\\	
%		\ker\Lambda_{\partial M}	
%				\ar[dd]^{d}\ar@{<->}[r]	&
%		d^{\star_{\partial M}}\Omega^2(\partial M)	\oplus H^1_{dR}(\partial M)
%				\ar[dd]^{\kappa\circ \delta}
%			&
%		H^1_{dR}(\partial M)	
%				\ar@{^{(}->}[l]		\ar@{-->}[d]^{\delta} 
%			&
%		\\
%			&
%			&
%		H^2_{dR}(M,\partial M)	
%				\ar[d]^{\kappa}
%			&
%		\\
%		dd^{\star_{\partial M}}\Omega^2(\partial M)
%			&
%		dd^{\star_{\partial M}}\Omega^2(\partial M)+ 		i^*_{\partial M}H^2_{dR}(M)
%				\ar@{->>}[l]
%			&
%		H^2_{dR}(M)	
%				\ar[l]
%	}\]
%	where the third columns is an exact sequence.

%%%%%%%%%%%%%%%

\begin{pro}\label{lma:C}
Let $0+G^0_{\partial M}\subset L_{\partial M}$ be the zero orbit for the gauge (free) action identified with the gauge group $G_{\partial M}^0$. Then
\[
	G^0_{\partial M}\subset L_{\tilde{M}}\cap L_{M,\partial M}^\bot.
\]

\end{pro}

\begin{proof}[Proof of Proposition \ref{lma:C}]
Take $f\in\Omega^0(\partial M)$, and $\psi_\geps:M\rightarrow[0,1]$ as in (\ref{eqn:psieps}).
If we define a function $\overline{f}:(\partial M)_\geps\rightarrow \mathbb{R}$ as
\begin{equation}\label{eqn:tilde f}
	\overline{f}(x,t):= \left(f\right)\circ (X_{\partial M}^{-\tau})(x,\tau),\qquad \forall x\in\Sigma,\,0\leq \tau \leq \geps
\end{equation}
where $X^{-\tau}_{\partial M}:=(X_{\partial M}^\tau)^{-1}$.Then $\overline{f}$ can be extended to $M$ via $\widetilde{f}:=\psi_\geps\cdot \overline{f}$. Furthermore $\gf\mapsto \gf+d\widetilde{f}$ describes an element of  $G_M$ such that $(d\widetilde{f})^D=df$ and $(d\widetilde{f})^N=0$, i.e. $d \overline{f}\in G_M$ is a section of the gauge homomorphisms $G_M^0\rightarrow G_{\partial M}^0$. Furthermore $d^\star dd\widetilde{f}=0$. Therefore $G_{\partial M}^0\subset L_{\tilde{M}}.$

From the very definition of ${L}_{M,\partial M}$  there is an inclusion $G_{\partial M}^0\subset {L}_{M,\partial M}^\bot:$
Take $\gf\in {L}_{M,\partial M}$, then for $(df,0)\in d\Omega^0(\partial M)^{\oplus 2}$, then the coisotropy condition reads as:
\[
	\int_{\partial M}\gf^D\wedge\star_{\partial M} d^{\star_{\partial M}}0-\int_{\partial M}df\wedge\star_{\partial M} \gf^N=
	-\int_{\partial M}f\wedge\star_{\partial M} d^{\star_{\partial M}}\gf^N=0.
\]
\end{proof}

\begin{cor}
Since $L_{M,\partial M}$ is coisotropic and $\ker\widetilde{\omega}_\Sigma\simeq G_{\Sigma}^0$, then $G_{\partial M}^0=G_{\partial M}^0\cap L_{M,\partial M} = L_{M,\partial M}^{\bot}.$
\end{cor}

Thus we have the following linear inclusions
\begin{equation}\label{eqn:*}
	\mathbf{L}_{\tilde{M}}= L_{\tilde{M}}\cap {L}_{M,\partial M}/ G^0_{\partial M}
		\subset \mathsf{L}_{\tilde{M}}\subset \mathsf{L}_{\partial M}
\end{equation}
\[
	\mathbf{L}_{M,\partial M}={L}_{M,\partial M}/{L}_{M,\partial M}^\bot\simeq 
	{L}_{M,\partial M}/G_{\partial M}^0\subset \mathsf{L}_{\partial M}.
\]

Recall that there is an exact sequence
\[\xymatrix{
	d\Omega^0_D(M)
			\ar@{^{(}->}[r]	&
	G_M^0
			\ar@{->>}[r]	&
	G_{\partial M}^0
			\ar@/_0.3cm/@{^{(}-->}[l]
}.\]
There is also an excision given by the map, $df\mapsto d\tilde{f}$, defined in (\ref{eqn:tilde f}). Hence there is a
well defined map $\mathsf{r}_M:L_M/G_M\rightarrow L_{\tilde{ M}}/G_{\partial M}^0$, whose image is $	\mathsf{r}_M(\mathsf{L}_{ M})\subset \mathsf{L}_{\tilde{M}}$ where $\mathsf{L}_{\tilde{M}}:={L}_{\tilde{M}}/ G^0_{\partial M} $.

For the proof of the following claim use the HMF decomposition on $M$ and the Hodge decomposition in $\partial M$.

\begin{pro}
We have the isomorphisms: a) $\Phi_{A_M}\simeq \mathsf{L}_M$; b) $\Phi_{A_{\partial M}}\simeq \mathsf{L}_{{\partial M}}$; c) $L_{\tilde{M}}\cap\Phi_{A_{\partial M}}\simeq L_{\tilde{M}}/{G^0}_{\partial M}$. So that the following diagram commutes
\begin{equation}\xymatrix{
	\mathsf{L}_M				\ar[r]^{\mathsf{r}_M}	\ar@{<->}[d]	&
	\mathsf{L}_{\tilde{M}}	\ar@{^{(}->}[r]		\ar@{<->}[d]&
	\mathsf{L}_{\partial M}						\ar@{<->}[d]
\\
	\Phi_{A_M}				\ar[r]_{r_M}			&
	L_{\tilde{M}}\cap\Phi_{A_{\partial M}}			\ar@{^{(}->}[r]		&
	\Phi_{A_{\partial M}}
}.\end{equation}

\end{pro}

Our previous discussion can be resumed in the following result about the symplectic framework for reduced abelian gauge field theories.

\begin{tma}\label{tma:2}
Consider the linear maps
\begin{equation}\label{eqn:diag1}\xymatrix{
	L_M			\ar@{-->>}[r]^{r_M}	\ar@{-->>}[d]^{\cdot/G^0_{\partial M}}		 	&
	L_{\tilde{M}}	\ar@{^{(}->}[r]	\ar@{->>}[d]^{\cdot/G^0_{\partial M}}	&
	L_{\partial M}	\ar@{->>}[d]^{\cdot/G^0_{\partial M}}
	\\
	\mathsf{L}_M				\ar@{->>}[r]^{\mathsf{r}_M}	&
	\mathsf{L}_{\tilde{M}}		\ar@{^{(}->}[r]		&
	\mathsf{L}_{\partial M}
	\\
					&
	\mathbf{L}_{\tilde{M}}	\ar@{^{(}->}[u]	\ar@{^{(}->}[r]	&
	\mathbf{L}_{M,\partial M} \ar@{^{(}->}[u]
}.\end{equation}
The following are true:

\begin{enumerate}
\item\label{tma2:part1} The squares of solid arrows commute.

\item\label{inciso:parte2} The image of the inclusion $\mathsf{r}_M(\mathsf{L}_{{M}})\subset \mathsf{L}_{\partial M}$ is isomorphic to the image of the inclusion (\ref{eqn:*}) of $\mathbf{L}_{\tilde{M}}$ as subspace of $\mathsf{L}_{\partial M}$.

\item\label{tma2:part3} The spaces $L_{\tilde{M}},\mathsf{L}_{\tilde{M}},\mathbf{L}_{\tilde{M}},$ in the middle column, are Lagrangian spaces contained into (pre)symplectic spaces  $L_{\partial M},\mathsf{L}_{\partial M},\mathbf{L}_{M,\partial M},$ respectively.

\end{enumerate}
\end{tma}

Recall that $L_{\partial M}$ is just presymplectic (coisotropic). We consider the definition of Lagrangian subspaces as subspaces of coisotropic spaces, see \cite{CMR}.

\begin{proof}[Proof of Theorem \ref{tma:2}]

Part \ref{tma2:part1} has already been shown. Part \ref{tma2:part3} is proved independently in \cite{CMR1} and \cite{DM}, see also Theorem \ref{tma:J} below. We prove part \ref{inciso:parte2}.

Since $\Phi_{A_M}=L_{\tilde{M}}\cap {L}_{M,\partial M},$ then the following diagram commutes
\begin{equation}\label{eqn:diag11}
\xymatrix{
	\mathsf{L}_{\tilde{M}}\cap \mathbf{L}_{M,\partial M}
		\ar@{^{(}->}[r]
			&
	\mathsf{L}_{\tilde{M}}
			\\
	\Phi_{A_M}
		\ar@{<->}[r]	\ar@{->>}[u]^{\cdot/G_{\partial M}^0}
			&	
	\mathsf{L}_M
		\ar@{->>}[u]
}.\end{equation}
Take $\phi\in L_{\partial M}$ and suppose that $\gf\in L_M$ is a solution with $r_M(\gf)=\phi=(\phi^D,\phi^N)$. Take its HMF decomposition
\[
	\gf=\omega+d^\star\alpha+d\beta
		\in\mathfrak{H}_N^1(M)\oplus d^\star\Omega^2_N(M)\oplus (\mathfrak{H}^1(M)\cap d\Omega^0(M)).
\]
Then $\gf_0:=d^\star\alpha$ solves the BVP
\[\left\{\begin{array}{ll}
	d^\star d\gf_0=0,
		&
	\\
	i^*_{\partial M}\gf_0=i^*_{\partial M}(d^\star\alpha),
		&
	\gf_0^N=\phi^N
\end{array}\right.\]
Notice that $\gf^N=(\omega+d^\star\alpha+d\beta)^N=(d^\star\alpha)^N=\phi^N$. So $\gf_0^N=\phi^N$.
On the other hand we can solve the following BVP for any Dirichlet boundary condition $i^*_{\partial M}(\omega)\in i^*_{\partial M}\mathfrak{H}_N^1(M)$, see \cite{BS},
\[\left\{\begin{array}{ll}
	\Delta\gf_1=0,
		&
	d^\star\gf_1=0
	\\
	i^*_{\partial M}\gf_1=i^*_{\partial M}(\omega),
		&
	\gf_1^N=0
\end{array}\right.\]
in particular $d^\star d\gf_1=0$, hence, $\gf_1\in L_M$. 

If $\gf_2:=\gf_0+\gf_1$ then $\gf_2\in \Phi_{A_M}$. Moreover for the Dirichlet condition
\[
	\gf_2^D=\gf_0^D+\gf_1^D=
	i^*_{\partial M}(d^\star\alpha)+i^*_{\partial M}(\omega)=\gf^D-di^*_{\partial M}\beta.
\]
Notice that modulo gauge, $\gf^D\sim\gf^D_2$. Meanwhile for the Neumann condition we have
\[
	\gf_2^N=\left(\gf_0+\gf_1\right)^N=\gf_0^N+\gf_1^N=\gf^N.
\]
Hence the linear map $\gf\mapsto \gf_2$ defines a projection $L_{{M}}\rightarrow \Phi_{A_M}$. The following diagram is commutative for the projections
\[\xymatrix{
	L_M
	\ar@{-->>}[r]		\ar@{->>}[d]
	&
	\Phi_{A_M}
	\ar@{=}[d]		\ar@{^{(}->}@/_/[l]
	\\
	L_{\tilde{M}}
	\ar@{-->>}[r]
	&
	L_{\tilde{M}}\cap {L}_{M,\partial M}
	\ar@{^{(}->}@/_/[l]
}.\]
Therefore we can design the following commuting diagram complementing (\ref{eqn:diag11}).
\[\xymatrix{
	L_{\tilde{M}}\cap {L}_{M,\partial M}
		\ar@{->>}[r]	
			&
	\mathbf{L}_{\tilde{M}}
		\ar@{^{(}->}[r]
			&
	\mathsf{L}_{\tilde{M}}\cap \mathbf{L}_{M,\partial M}
		\ar@{^{(}->}[r]
			&
			\mathsf{L}_{\tilde{M}}
			\\
			&
	L_{{M}}
		\ar@{-->>}[r]	\ar@{->>}[lu]
			&
	\Phi_{A_M}
		\ar@{->>}[u]^{\cdot/G^0_{\partial M}}
		\ar@{<->}[r]
		&
		\mathsf{L}_M
		\ar@{->>}[u]
}\]
Recall Proposition \ref{lma:(d)}. This proves that  the image, $\mathsf{r}_M(\mathsf{L}_M)\subset \mathsf{L}_{\tilde{M}}\cap\mathbf{L}_{M,\partial M},$ equals the image of the inclusion $\mathbf{L}_{\tilde{M}}\subset \mathsf{L}_{\tilde{M}}\cap\mathbf{L}_{M,\partial M}$, therefore $\mathbf{L}_{\tilde{M}}\simeq \mathsf{r}_M(\mathsf{L}_M)$.			\end{proof}

%The following Lemma will be useful for the proof of Theorem \ref{pro:1}.

\begin{lma}\label{lma:extension}
Take $\phi\in\Omega^k(\partial M)$, then there exists an extension, $\widetilde{\phi}\in\Omega^k(M)$, such that
\[
	i^*_{\partial M}\widetilde{\phi}=\phi,\qquad i^*_{\partial M}\left(\star \widetilde{\phi}\right)=0,\qquad
	 i^*_{\partial M}\left(\star d\widetilde{\phi}\right)=0.
\]
in particular $\widetilde{\phi}\in\Omega_N^k(M)$.
\end{lma}
\begin{proof}
Define $\overline{\phi}$ as in (\ref{eqn:overphi}). This $\overline{\phi}\in\Omega^1(\partial M_\geps)$ can be used to define an extension in $M$ as $\widetilde{\phi}:=\psi_\geps\cdot \overline{\phi}$, where $\psi_\geps$ was defined in (\ref{eqn:psieps}).
Then $i^*_{\partial M}\overline{\phi}=\phi$ and also
\begin{equation}\label{eqn:locexpresion}
	i^*_{\partial M}\left(\star\widetilde{\phi}\right)=(X_{\partial M}^0)^*\left(\star \overline{\phi}(x)\right)=
	\sum_{I}\phi_I(x)i^*_{\partial M}(\star d x^I)=0,
\end{equation}
Finally
\[
	d\widetilde{\phi}\mid_{\partial M}=d\left(\psi_\geps\cdot \overline{\phi}\right)\mid_{\partial M}=
	\psi_\geps\cdot d\left(\overline{\phi}\right)\mid_{\partial M}=
	 d{\phi}\mid_{\partial M}
\]
Hence $i^*_{\partial M}\left(\star d{\widetilde{\phi}}\right)	 =i^*_{\partial M}(\star d{\phi})=0$, since we can obtain local expressions similar to those in (\ref{eqn:locexpresion}) for $i^*_{\partial M}\left(\star\widetilde{\phi}\right)$.	\end{proof}

\section{Gluing}\label{sec:3}

Suppose that a region $M_1$ is obtained from a primitive region $M$ by gluing along $\Sigma,\Sigma'\subset \partial M$. Then $\partial M_1\subset \partial M$. There is a commuting diagram of linear maps
\begin{equation}\label{eqn:gluing1}\xymatrix{
	\mathsf{L}_{M_1}
		\ar[rrr]		\ar[dr]	\ar@{-->}[dd]
		&&&
	\mathsf{L}_{M}
		\ar[dl]			\ar@{-->}[dd]
	\\
	&
	\mathbf{L}_{M_1,\partial M_1}
		\ar@^{{(}->}[dl]	 
	&
	\mathbf{L}_{M,\partial M}
		\ar@^{{(}->}[dr]		\ar@{->>}[l]
	&
	\\
	\mathsf{L}_{\partial M_1}
	&&&
	\mathsf{L}_{\partial M}
	\ar@{^{(}->}[lll]
}.\end{equation}

\begin{tma}\label{cor:tma1}
The reduced space $\mathbf{L}_{M,\partial M}\subset \mathsf{L}_{\partial M}$ has finite codimension.
\end{tma}

Intuitively, the gluing process increases the topological manifestation of the homology of the interior from the point of view of the boundary. This can be formalized as an inequality that shows a monotone decreasing of the codimension mentioned in Theorem \ref{cor:tma1} as a sequence of gluings is applied. Namely we have:
\begin{equation}\label{eqn:gluingcodimension}
\mathrm{codim \,} \mathbf{L}_{M_1,\partial M_1}\leq\mathrm{codim \,} \mathbf{L}_{M,\partial M}
\end{equation}

The following Lemma will be crucial for the proof of Theorem \ref{cor:tma1}.

\begin{lma}\label{lma:(b)}
Suppose that $\tilde{\chi}\in\Omega^1(M)$ satisfies the conclusion of Lemma \ref{lma:extension}, namely
\[
 	i^*_{\partial M}\left(\star \widetilde{\chi}\right)=0,\qquad
	 i^*_{\partial M}\left(\star d\widetilde{\chi}\right)=0.
\]
Then
 \[
 	d^{\star_{\partial M}}\left(i^*_{\partial M}\widetilde{\chi}\right)=
	 i^*_{\partial M}\left(d^\star\widetilde{\chi}\right)=
 	i^*_{\partial M}\left(d^\star\widetilde{\chi}\mid_{\partial M}\right).
\]
\end{lma}

\begin{proof}

In local expressions, $\widetilde{\chi}\mid_{\partial M}=\sum_{j=1}^{n-1}\chi_j(x)dx^j$, and both hypothesis read as ${\chi}_{\tau}(x,0)=0$ and  $\partial_\tau {\chi}_j(x,0)=0$, with $x\in\partial M$, for $j=1,\dots,n-1$, respectively. Here $x^n=\tau$ is the normal coordinate. In fact they come from $i^*_{\partial M}\left(\star d\widetilde{\chi}\right)=0$ with
\[
	i^*_{\partial M}\left(\star d\widetilde{\chi}\right)=
	\sum_{j=1}^{n-1}\partial_{\tau}\widetilde{\chi}_j(x,0)\cdot i^*_{\partial M}\left(\star d\tau\wedge dx^j\right).
\]
Hence by  local calculations, see for instance \cite{Jost-RG},
\[
	d^{\star}\widetilde{\chi}=-\sum_{k,l=1}^{n}\left({h^{kl}}\frac{\partial \chi_k}{\partial x^l}-\sum_{r=1}^n\Gamma^r_{lk}\cdot \chi_r\right)
	=
	\sum_{k,l=1}^{n-1}
	\left({\overline{h}^{kl}}\frac{\partial \chi_k}{\partial x^l}-
	\sum_{r=1}^{n-1}\overline{\Gamma}^r_{lk}\cdot \chi_r\right)=
	\]
	\[
	=d^{\star_{\partial M}}{\chi},\qquad \chi:=i^*_{\partial M}\widetilde{\chi}.
\]
where $\Gamma^r_{lk}$ denote the Christoffell symbols of the Riemannian metric on $M$, while $\overline{\Gamma}^r_{lk}$ denote those of the induced Riemannian metric $\overline{h}$ on $\partial M$.

\end{proof}

%%%%%

%%%%%%%%%%%%%%%%%%%%%%%%

\begin{proof}[Proof of Theorem \ref{cor:tma1}]
First recall that $G_{\partial M}^0$ acts trivially in the Neumann boundary conditions, hence
\[
	{L}_{M,\partial M}^N=
	{L}_{M,\partial M}^N/G_{\partial M}^0\simeq
	\mathbf{L}_{M,\partial M}^N\subset\ker d^{\star_{\partial M}}
\]
where
\[
	{L}_{M,\partial M}^N:=\left\{0\oplus\gf^N:\, \gf\in\mathfrak{H}_N^1(M)\oplus d^\star\Omega_N^2(M) \right\}.
\]
For Dirichlet boundary conditions, define
\[
	{L}_{M,\partial M}^D:=
		i^*_{\partial M}\left(\mathfrak{H}_N^1(M)\oplus d^\star\Omega_N^2(M) \right)
\]
take $\gf=\kappa+ d^{\star}\alpha\in\mathfrak{H}_N^1(M)\oplus d^\star\Omega_N^2(M)$. According to \cite{BS} 
\[
	i^*_{\partial M}\kappa\in
		i^*_{\partial M}\mathfrak{H}_N^1( M)
		\subset
		i^*_{\partial M}\mathfrak{H}^1( M)
		\subset \mathfrak{H}^1(\partial M)\oplus d\Omega^0(\partial M)
\]
therefore, the inclusion
\[
	\left(i^*_{\partial M}\mathfrak{H}_N^1( M) \right)/d\Omega^0(\partial M)
	\subset
	\mathfrak{H}^1(\partial M) 
\]
has finite dimension, in the finite dimensional space $\mathfrak{H}^1(\partial M)$.

On the other hand, for every $\chi\in\Omega^2(\partial M)$, there exists an extension $\widetilde{\chi}$, such that
\[
	 i^*_{\partial M}\left(\star \widetilde{\chi}\right)=0=i^*_{\partial M}\left(\star d\widetilde{\chi}\right)=0,\,
	 i^*_{\partial M}\widetilde{\chi}=\chi,
\]
and
\[
	d^{\star_{\partial M}}\chi=i^*_{\partial M}\left(d^\star\widetilde{ \chi}\right)
	=	d^\star\widetilde{\chi}\mid_{\partial M}
\]
see Lemmas \ref{lma:extension} and \ref{lma:(b)}. Hence $ d^{\star_{\partial M}}\Omega^2(\partial M)\subset i^*_{\partial M}\left(d^\star\Omega_N^2(M)\right)$.
Therefore when we consider the action of $G_{\partial M}^0$ on $i^*_{\partial M}\left(d^\star\Omega_N^2(M)\right)$, we have
\[
{L}_{M,\partial M}^D/G_{\partial M}^0 \subset 
\mathfrak{H}^1(\partial M)\oplus d^{\star_{\partial M}}\Omega^2(\partial M)
=\mathsf{L}_{\partial M}^D
\]
has finite codimension.		\end{proof}

%%%%%%%%%%%%%%%%

%%%%%%%%%%%%%%%%%%%%%%%%%%%%%%%%%%%%%%%%%%%%%%%%%%%%%%%%%%%

\section{Complex structure for boundary conditions}\label{sec:J-structure}

%%%%%%%%%%%%%%%%%%%%%%%%%%%%%%%%%%%%%%%%%%%%%%%%%%%%%%%%%%%

We claim that the Dirichlet to Neumann operator yields a complex structure for the space of boundary conditions. We consider a space-time region $M$ that is a Riemannian smooth manifold with (smooth) boundary $\partial M$.
%%%%%%%%%%%%%%%%%%%%%%%%%%%%%%%%%%%%

\subsection{Dirichlet to Neumann operator on $k-$forms}

%%%%%%%%%%%%%%%%%%%%%%%

For $k-$forms several proposals have been explored, see references in \cite{BS}. Recall that  every solution of (\ref{eqn:BVP2}) is also a solution to the Yang-Mills BVP given in (\ref{eqn:BVP3-YM}). Every solution of (\ref{eqn:BVP3-YM}) induces in turn a solution to the following less restrictive Yang-Mills BVP
\begin{equation}\label{eqn:BVP2-YM}
\left\{\begin{array}{ll}
	d^\star d\gf=0, & 
	\\
	i_{\partial M}^*\gf=\phi^D, & i^*_{\partial M}(d^{\star}\gf)=0
\end{array}.\right.
\end{equation}

We define the \emph{Dirichlet to Neumann operator associated to the region} $M$ and to the BVP (\ref{eqn:BVP2-YM}) as
\begin{equation}\label{eqn:Lambda}
	{\Lambda_{\tilde{M}}}\left(\phi^D\right):=(-1)^{kn}\star_{\partial M}i_{\partial M}^*\left(\star d\gf\right).
\end{equation}
Remark that we adopt the convention of Dirichlet to Neumann operator ${\Lambda_{\tilde{M}}}:\Omega^k(\partial M)\rightarrow\Omega^k(\partial M)$, instead of $\Lambda:\Omega^k(\partial M)\rightarrow\Omega^{n-k}(\partial M)$ given in \cite{BS} and references therein. The motivation for this choice is to consider the graph of this operator contained in a tangent space $T\Omega^k(\partial M)$, rather than contained in the cotangent space $T^*\Omega^k(\partial M)$. This is consistent with our Lagrangian approach rather than with a Hamiltonian framework for gauge fields. The D-N operator $\Lambda_{\tilde{M}}$ is a closed, positive definite one, see \cite{BS}.

In particular, if we consider a solution $\gf$ whose boundary condition has no Neumann component, $i^*_{\partial M}\left(\star d\gf\right)=0$, then $\gf^D\in\ker\Lambda_{\tilde{M}}$. Hence by Lemma \ref{lma:(b)},
$
	i_{\partial M}^*(d^{\star}\gf)=
	d^{\star_{\partial M}}\phi^D.
$
The boundary condition, $i_{\partial M}^*(d^{\star}\gf)=0$, implies
\begin{equation}\label{eqn:ker-Lambda}
	\ker\Lambda_{\tilde{M}}\subset		\ker d^{\star_{\partial M}}.
\end{equation}

The proof of the following result follows Lemma 3.2 in \cite{BS}.

\begin{lma}\label{lma:previous}
There exists an isomorphism $j_{\tilde{M}}:\ker \Lambda_{\tilde{M}}\rightarrow\ran\Lambda_{\tilde{M}},$ defined as the composition of linear maps $j_{\tilde{M}}= (j^N)^{-1}\circ \star_{\partial M}\circ j^D$, where:

\[
	j^D:\ker \Lambda_{\tilde{M}}
		\rightarrow
	\begin{array}{l}
	i^*_{\partial M}\left(
				{\mathfrak{H}}^k_N(M)
				\oplus	d^\star\Omega_N^{k+1}(M)\oplus	
		\right.\\\left.
				\oplus
				\left(
					{\mathfrak{H}}^{k}(M)\cap d\Omega^{k-1}(M)
				\right)
			\right)	
			\cap \ker d^{\star_{\partial M}},
	\end{array}
\]
\[j^N:
			\ran\Lambda_{\tilde{M}}
	\rightarrow
		\begin{array}{l}
			i^*_{\partial M}\left(
				\mathfrak{H}^{n-1-k}_D(M)
				\oplus	d\Omega^{n-2-k}_D(M)\oplus
		\right.\\\left.
				\oplus	\left(
					\mathfrak{H}^{n-1-k}(M)\cap d\Omega^{n-2-k}(M)
					\right)
			\right)\cap \ker d^{\partial M}.
		\end{array}
\]
\end{lma}
In fact since the $G^0_{\partial M}-$action acts trivially in $\ker d^{\star_{\partial M}}\subset L^D_{M,\partial M}$, we have the inclusion 
\[
	j_D:\ker \Lambda_{\tilde{M}}\rightarrow L^D_{M,\partial M}\cap \ker d^{\star_{\partial M}}\simeq\mathsf{L}^D_{M,\partial M}
\]
where $\mathsf{L}^D_{M,\partial M}:= {L}_{M,\partial M}^D/G^0_{\partial M}.$

\begin{proof}[Proof of Lemma \ref{lma:previous}]

Define the map $j^D(\gf):=\gf^D=i^*_{\partial M}\gf$, where $\gf$ is the solution  with $\gf^N=0$ of the BVP (\ref{eqn:BVP2-YM}). According to the HFM decomposition, we have $	\gf=\psi+\rho,$ where
\[
		\psi\in \mathfrak{H}^k_D(M)\oplus d\Omega^{k-1}_D(M),
\]
\[
		\rho\in \mathfrak{H}^k_N(M)\oplus 
		d^\star\Omega_N^{k+1}(M)\oplus\left(\mathfrak{H}^k(M)\cap d\Omega^{k-1}(M)\right).
\]
Notice that $\rho^D=\gf^D$. Consider the following BVP
\begin{equation}\label{eqn:BVP-lambda}
	\left\{\begin{array}{ll}
		d\lambda=0,\\
		i^*_{\partial M}\lambda=\star_{\partial M}\gf^D
	\end{array}\right.
\end{equation}
we claim that if $d^{\star_{\partial M}}\gf^D=0$, or equivalently $d(\star_{\partial M}\gf^D)=0$, then there exists a solution 
\[
	\lambda\in
		\mathfrak{H}^{n-1-k}_D(M)\oplus d\Omega^{n-2-k}_D(M)\oplus
		\left(\mathfrak{H}^{n-1-k}(M)\cap d\Omega^{n-2-k}(M)\right)
\]
to (\ref{eqn:BVP-lambda}). For this, define $\lambda=\widetilde{\star_{\partial M}\gf^D}$.
Hence $	\star\lambda=d\mu+\gamma$, where
\[
	\gamma \in \mathfrak{H}^{k+1}_N(M)\oplus d^\star\Omega_N^k(M),\qquad
	d\mu\in\left(\mathfrak{H}^{k+1}(M)\cap d\Omega^{k}(M)\right).
\]
Notice that $d\mu$ is harmonic in $M$, so $d^\star d\mu=0$. We claim that
\[
	(j^N)^{-1}\left(i^*_{\partial M}(\lambda)\right)=(-1)^{nk}\star_{\partial M}i^*_{\partial M}\left(\star d\mu\right)\in \ran{\Lambda_{\tilde{M}}}.
\]
The equality $\star\star\lambda=\star d\mu+\star\gamma$ implies that
\[
	(-1)^{(n-1-k)(1+k)}i_{\partial M}^*\lambda=i^*_{\partial M}\left(\star\star \lambda\right)=
	i^*_{\partial M}\left(\star d\mu+\star\gamma\right)=
\]
\[
	=i^*_{\partial M}\left(\star d\mu\right)
\]
since $\star\gamma\in \Omega_D(M)$. Hence
\[
	\star_{\partial M}i^*_{\partial M}\left(\star d\mu\right)=\]
\[
	(-1)^{(n-1-k)(1+k)}\star_{\partial M}i_{\partial M}^*\lambda
	=(-1)^{(n-1-k)(1+k)}\star_{\partial M}\star_{\partial M}\gf^D
\]
\[
	=(-1)^{(n-1-k)(1+k)}\cdot(-1)^{k(n-1-k)}\cdot\gf^D=(-1)^{n-1-k}\gf^D
\]
Thus
$(-1)^{nk}\star_{\partial M}i^*_{\partial M}\left(\star d\mu\right)=(-1)^{nk}\cdot(-1)^{n-1-k}\cdot\gf^D$.
Therefore
$$
	(j^{N})^{-1}\left(i^*_{\partial M}(\lambda)\right)=(-1)^{(n-1)(k+1)}\gf^D
$$
or
\[
	(-1)^{(n-1-k)(1+k)}\cdot\left[(j^N)^{-1}\circ\left(\star_{\partial M} j^D\right)\right]\left(\gf^D\right)
	=(-1)^{(n-1)(k+1)}\phi^D
\]
\begin{equation}\label{eqn:jN^-1}
	\left[(j^N)^{-1}\circ\left(\star_{\partial M} j^D\right)\right]\left(\gf^D\right)
	=  	(-1)^{k(1+k)}\cdot \gf^D.
\end{equation}
By gauge choice we can consider a solution $\mu$ such that $d^\star d\mu=0$, but also $i^*_{\partial M}\left(d^*\mu\right)=0$. This solves (\ref{eqn:BVP2-YM}).
\end{proof}

Notice that for $1-$forms, $k=1$, we have $j_{\tilde{M}}(\gf^D)=\gf^D$.
%%%%%%%%%%%%%%%%%%%%%%%%%%%

\subsection{Tame complex structure}

%%%%%%%%%%%%%%%%%%%%%%%%%%%

By (\ref{eqn:ker-Lambda}) there exists an inclusion
\begin{equation}\label{eqn:ran-Lambda}
	\star_{\partial M}\circ j^N:\ran\Lambda_{\tilde{M}}\rightarrow \ker d^{\star_{\partial M}}.
\end{equation}
Furthermore we can define the operator $J$ on $\ker\Lambda_{\tilde{M}}\oplus\ran\Lambda_{\tilde{M}}$, as
\begin{equation}\label{eqn:J}
	\left(\begin{array}{cc}
		0&- \left(j_{\tilde{M}}\right)^{-1} \\
		 j_{\tilde{M}}& 0
	\end{array}\right)
	:\ker{\Lambda_{\tilde{M}}}\oplus\ran{\Lambda_{\tilde{M}}} 
		\rightarrow \ker{\Lambda_{\tilde{M}}}\oplus \ran{\Lambda_{\tilde{M}}}.
\end{equation}
This is a complex structure which is tame with respect to the symplectic structure $\omega_{\partial M}\mid_{L_J}$. There is a linear inclusion
\[
	L_J:=L_J^D\oplus L_J^N\subset \Phi_{A_{\partial M}}\cap L_{M,\partial M} \simeq \mathbf{L}_{M,\partial M}
\]
where
\[
	L_J^D:=j^D\left(\ker\Lambda_{\tilde{M}}\right)\subset L^D_{M,\partial M},\quad 
	L_J^N:=\star_{\partial M}\circ j^N\left(\ran\Lambda_{\tilde{M}}\right)\subset L^N_{M,\partial M}.
\]
The taming condition is
\[
	g_{\partial M}\mid_{L_J}(\cdot,\cdot)=2\,\omega_{\partial M}\mid_{L_J}(\cdot,J\cdot),
\]
the bilinear form $g_{\partial M}$ can be explicitly calculated as
\[
	g_{\partial M}\left({\phi_1},{\phi_2}\right)=
		\int_{\partial M}\phi_1^D\wedge\star_{\partial M}\phi_2^D+ \phi_1^N\wedge\star_{\partial M}\phi_2^N
		\qquad \forall {\phi_i}\in {L}_{\partial M},\, i=1,2.
\]
that is positive definite in $\Phi_{\partial M}\subset L_{\partial M}$.
This allows us to prove the following result.

\begin{tma}\label{tma:J}
The following are true:
\begin{enumerate}[label=(\roman*)]

\item\label{th-part:2} There exists an isomorphism
\[
	\mathbf{L}_{M,\partial M}\simeq L_J\subset\mathsf{L}_{\partial M}.
\]

\item\label{th-part:1} The operator $J$ satisfies $J^2=-Id,$  hence $J$ is a complex structure $J:\mathbf{L}_{M,\partial M}\rightarrow \mathbf{L}_{M,\partial M}$.

\item $\mathbf{L}_{\partial M}$ is isomorphic to a symplectic subspace of the linear spaces $\mathsf{L}_{\partial M}$.

\item\label{th-part:3} The inclusion $\mathbf{L}_{\tilde{M}}\subset \mathbf{L}_{M,\partial M}$ is a graph (of a linear isomorphism).

\item\label{th-part:4} $\mathbf{L}_{M,\partial M}$ decomposes as a direct sum $\mathbf{L}_{\tilde{ M}}\oplus J\mathbf{L}_{\tilde{ M}}$.

\item\label{th-part:5}  $\mathsf{r}_M\left(\mathsf{L}_M\right)=\mathbf{L}_{\tilde{ M}}\subset\mathbf{L}_{M,\partial M}$ is a Lagrangian subspace.

\end{enumerate}
\end{tma}

\begin{proof}

Part \ref{th-part:1}  follows from our previous comments. Part \ref{th-part:3} follows form uniqueness of solutions of (\ref{eqn:BVP3-YM}) which in turn follows from uniqueness of solution to the corresponding BVP up to $\lambda\in\mathfrak{H}^1_D(M)$. Part \ref{th-part:4} follows from \ref{th-part:2} and \ref{th-part:3}. Part \ref{th-part:5} follows from \ref{th-part:3}, see \cite{CMR, CMR1}.

Assertion \ref{th-part:2} needs to be proven. Notice that $L_J^D\oplus L_J^N\subset (\ker d^{\star_{\partial M}})^{\oplus 2}$ hence there is an inclusion $L_J\subset \mathbf{L}_{M,\partial M}.$ Furthermore  $L_J$ is a symplectic space because it is a complex space for $J$ tame.

We just need to prove that $ L_J=\mathbf{L}_{M,\partial M}$. In fact we need $	\mathbf{L}_{M,\partial M}^D\subset L_J^D$
or
\begin{equation}\label{eqn:space2}
	\mathbf{L}_{M,\partial M}^D\subset
		i^*_{\partial M}\left( \ker \Lambda_{\tilde{M}}\right).
\end{equation}
Take $\gf=\omega+ d^\star\alpha\in \mathfrak{H}_N^1(M)\oplus d^{\star}\Omega_N(M)$. Recall that that
\[
	i^*_{\partial M}\mathfrak{H}_N^1(M)=\mathfrak{H}^1(\partial M)\oplus d\Omega^0(\partial M)
\]
hence $i^*_{\partial M}\omega =\lambda +d\gamma$. Recall also that there exists $\omega_1\in\Omega^1(M)$ solving the BVP
\[\left\{\begin{array}{ll}
	\Delta\omega_1=0,
	&
	d^\star\omega_1=0
	\\
	i^*_{\partial M}\left(\omega_1\right)=\lambda,
	&
	i^*_{\partial M}\left(\star d\omega_1\right)=0.
\end{array}\right.\]
Notice that $\omega_1$ also solves the Yang-Mills BVP
\[\left\{\begin{array}{ll}
	d^\star d \omega_1=0,
	&
	d^\star\omega_1=0
	\\
	i^*_{\partial M}\left(\omega_1\right)=\lambda,
	&
	i^*_{\partial M}\left(\star d\omega_1\right)=0.
\end{array}\right. \]
Also $i^*_{\partial M}(\omega_1)\in\ker \Lambda_{\tilde{M}}$.

Now we show that $\omega-\omega_1=0.$ For every $\gf_1\in {L}_{M,\partial M}$
\[
	\omega_{\partial M}\left(r_M(\omega-\omega_1),\,r_M(\gf_1)\right)=
	\int_{\partial M}(\omega-\omega_1)^D\wedge\star_{\partial M} \gf_1^N-
	\int_{\partial M}\gf^D\wedge\star_{\partial M}(\omega-\omega_1)^N.
\]
According to the results in \cite{BS}, $\omega_1$ is also harmonic, $\omega_1\in\mathfrak{H}_N^1(M)$. Hence $\omega_1^N=0$,
\[
	\omega_{\partial M}\left(r_M(\omega-\omega_1),\,r_M(\gf_1)\right)=
	\int_{\partial M}(\omega-\omega_1)^D\wedge\star_{\partial M}\gf_1^N=
\]
\[
	=\int_{\partial M}d\gamma\wedge\star_{\partial M}\gf_1^N
		=\int_{\partial M}\gamma\wedge\star_{\partial M}d^{\star_{\partial M}}\gf_1^N=0.
\]
Recall Proposition \ref{lma:(d)} for the last equality. Therefore, by non-degeneracy of the symplectic form $\omega_{\partial M}$, $i^*_{\partial M}(\omega)=i^*_{\partial M}(\omega_1)\in\ker\Lambda_{\tilde{M}}$.

\end{proof}

%%%%%%%%%%%%%%%%%%%%%%%

\subsection{Complex structure for hypersurface solutions}

%%%%%%%%%%%%%%%%%%%%%%%

Recall that for a hypersurface $\Sigma$ we have a cylinder $\Sigma_\geps$ provided with a Riemannian metric, $\partial \Sigma_\geps=\Sigma\sqcup \Sigma'$. We also have a diffeomorphism $X^\geps_{\Sigma}:\Sigma\rightarrow\overline{\Sigma'}$, where $\overline{\Sigma'}$ means reversed orientation with respect to the orientation on $\Sigma'$, this orientation is induced by the orientation in the interior, $\Sigma_\geps$. We prove explicit local existence results. These are rather well known arguments. General local existence results for non-abelian Yang-Mills fields may be found in \cite{MariniDN}, references therein deal with the abelian case.
we consider the \emph{D-N operator associated to the hypersurface} $\Sigma$ defined as:
\begin{equation}\label{eqn:Lambda-Sigma}
	{\Lambda}_{\tilde{\Sigma}_\geps}(\phi^D):=(-1)^{n}\star_{\Sigma}i_{\Sigma}^*\left(d\star_{\Sigma_\geps}\gf\right),
	\qquad
	\phi^D\in \ker d^{\star_\Sigma}\subset \Omega^1(\Sigma)
\end{equation}
where $\gf$ is a solution to (\ref{eqn:BVP-Sigma}), by considering $M=\Sigma_\geps$ and the associated D-N operator.

We have already shown in (\ref{eqn:ker-Lambda}) that $\ker\Lambda_{\tilde{\Sigma}\geps}\subset \ker d^{\star_{\Sigma}}$. We just need to prove that the inclusion is surjective. Consider a solution $\gf$ to the BVP (\ref{eqn:BVP2-YM}), we just need to prove that $i_\Sigma^*(\star d\gf)=0$.

\begin{lma}\label{lma:6}
We have the isomorphisms $
	\ker d^{\star_{\Sigma}}\simeq
		\ker{\Lambda}_{{\tilde{\Sigma}_\geps}}\simeq \ran{\Lambda}_{\tilde{\Sigma}_\geps}.$ 
\end{lma}

\begin{proof}Take the local expression
\[
	\gf=\sum_{i=1}^{n-1}\gf_idx^i+\gf_\tau d\tau\in\Omega^1(\Sigma_\geps),
\]
hence \cite{Jost-RG}
\[
	d^\star\gf=-\sum_{k,l=1}^n\left[
			h^{kl}\frac{\partial \gf_k}{\partial x^l}-\sum_{j=1}^n\Gamma_{kl}^j\gf_j
		\right]=
\]
\[
	=-h^{\tau\tau}\frac{\partial\gf_\tau}{\partial \tau} -\sum_{l=1}^{n-1}h^{\tau l}\frac{\partial\gf_\tau}{\partial x^l}-
		\sum_{k=1}^{n-1}h^{k\tau}\frac{\partial\gf_k}{\partial \tau}+
\]
\[		+ \sum_{l=1}^{n-1}\sum_{j=1}^n\Gamma_{\tau l}^j\gf_j+\sum_{k=1}^{n-1}\sum_{j=1}^n\Gamma^j_{k\tau}\gf_ j+		\sum_{j=1}^n			\Gamma^j_{\tau\tau}\gf_j
\]
\[
		-\sum_{j,k=1}^{n-1}\left[
			\overline{h}^{kl}\frac{\partial \gf_k}{\partial x^l}-\sum_{j=1}^{n-1}\overline{\Gamma}_{kl}^j\gf_j
		\right]
\]
where we consider Christoffel symbols. Because of the orthogonality condition $g^{nk}(0)=\delta_{n,k}$ the Kronecker delta for $k=1,\dots,n-1,n$. Since $\tau$ is the time parameter for geodesics, hence the Christoffel symbols with $\tau$ index vanish $\Gamma_{\tau\cdot}^\cdot=0=\Gamma_{\cdot\tau}^\cdot,$ see \cite{Michor}. We also have orthonormality along the geodesic so
\[
	h^{\tau\tau}(s,\tau)=1,\, h^{\tau i}(s_0,\tau)=h^{i\tau}(s_0,\tau)=0,\qquad i=1,\dots, n-1.
\]
Therefore, we have a simplified local expression for the divergence $d^\star\gf$,
\[
		d^\star\gf (s_0,\tau)=
			-\frac{\partial\gf_\tau}{\partial \tau} 
			-\sum_{j,k=1}^{n-1}\left[
			\overline{h}^{kl}\frac{\partial \gf_k}{\partial x^l}-\sum_{j=1}^{n-1}\overline{\Gamma}_{kl}^j\gf_j
		\right]
\]
or
\begin{equation}
		d^\star\gf (s_0,\tau)=
		-\frac{\partial\gf_\tau}{\partial \tau} +d^{\star_{\Sigma^{\tau}}}\left(i^*_{\Sigma^{\tau}}\gf\right)
\end{equation}
where $\Sigma^{\tau}:=X^{\tau}_\Sigma(\Sigma)\subset \Sigma_\geps$, $(s,\tau)\in \Sigma\times[0,\geps]$. Here $i^*_{\Sigma^\tau}:=X^{-\tau}_\Sigma:\Sigma^{\tau}\rightarrow\Sigma$. This equation remains valid along the geodesic $\gamma_{s}(\tau)$, for every $s\in\Sigma$. 

The condition $d^\star\gf=0$ can be achieved once we solve the ODE for every fixed initial condition in $s\in\Sigma$
\begin{equation}\label{eqn:ODE1}
\left\{\begin{array}{c}
	\frac{\partial\gf_\tau}{\partial \tau} =d^{\star_{\Sigma^{\tau}}}{\eta^\tau}\\
	\gf_\tau(s,0)=\gf^0(s)
\end{array}\right.
\end{equation}
where $\eta^{\tau}:=i^*_{\Sigma^{\tau}}\gf\in\Omega^1(\Sigma^\tau),$ and $\gf^\tau(s):=\gf_\tau(s,\tau)\in\Omega^0(\Sigma),\,\forall \tau\in[0,\geps].$ 

On the other hand by Cartan's Formula, $\mathcal{L}_{\partial_\tau}\gf= \partial_\tau\llcorner (d\gf)+d\left(\partial_\tau\llcorner \gf\right)$, locally
\[
	\mathcal{L}_{\partial_\tau}\gf=
		\left[ \sum_{j=1}^{n-1}(\partial_\tau\gf_j-\partial_j\gf_\tau)dx^j\right]
	+
		\left[\sum_{j=1}^{n-1}(\partial_j\gf_\tau)dx^j+\partial_\tau\gf_\tau d\tau\right]
		\]
		\[
	= \partial_\tau\llcorner d\gf+
		d{\gf^\tau}+
		\partial_\tau\gf_\tau d\tau
\]
or
\begin{equation}\label{eqn:XXX}
	\mathcal{L}_{\partial_\tau}\gf=	 \partial_\tau\llcorner d\gf+
		d{\gf^\tau}+
		\left(d^{\star_{\Sigma^\tau}}\eta^\tau\right)d\tau
\end{equation}
where we use (\ref{eqn:ODE1}). Recall that
\[
	i^*_{\Sigma^{\tau}}\left(\star d\gf\right)
		=i^*_{\Sigma^{\tau}}\left(\star\sum_{j=1}^{n-1}(\partial_\tau\gf_j )\,d\tau\wedge dx^j\right)+
		i^*_{\Sigma^\tau}\left(\star\sum_{j=1}^{n-1}(\partial_j\gf_\tau) \,dx^j\wedge d\tau\right)
		=\]\[
		= \sum_{j=1}^{n-1}\partial_\tau\gf_j(\star_{\Sigma^\tau} dx^j)-
		 \sum_{j=1}^{n-1}\partial_j\gf_\tau(\star_{\Sigma^\tau} dx^j).
\]
If we apply $\star_{\Sigma^\tau} i^*_{\Sigma^{\tau}} $ to (\ref{eqn:XXX}) then we have the relation
\begin{equation}\label{eqn:XX}
\star_{\Sigma^\tau} i^*_{\Sigma^\tau}\left(\mathcal{L}_{\partial_\tau}\gf\right)
	=		i^*_{\Sigma^\tau}(\star d\gf)
	+		\star_{\Sigma^\tau}(d\gf^\tau).
\end{equation}
Solving the $(n-1)-$dimensional ODE
\begin{equation}\label{eqn:ODE2}
\left\{\begin{array}{c}
	\frac{\partial \eta^{\tau}}{\partial\tau}=d\gf^\tau
	\\
	\eta^\tau(s,0)=\phi^D(s)
\end{array}\right.
\end{equation}
is equivalent to $ i^*_{\Sigma^\tau}\left(\mathcal{L}_{\partial_\tau}\gf\right)=(d\gf^\tau)$. Hence by (\ref{eqn:XX}) we have $i^*_{\Sigma^\tau}\left(\star d\gf\right)=0$.

When we take the differential $d$ in (\ref{eqn:XX}) we have
\[
	i^*_{\Sigma^\tau}(d\star d\gf)=di^*_{\Sigma^\tau}(\star d\gf)=0.
\]
In order to find a solution to the ODE  (\ref{eqn:ODE2}) we just need to prescribe a differentiable $1-$parameter family, $\eta^\tau\in\Omega^1(\Sigma)$. Its velocity should be constant and exact, in fact $\partial_\tau\eta^\tau=df$. Thus if we define $\gf^\tau=\gf^0=f$, then we have suitable initial conditions for solving this ODE. Once we have found $\eta^\tau$ we solve the ODE (\ref{eqn:ODE1}) to obtain $\gf_\tau$.

Recall that the space of exact forms is transverse to the space of $d^{\star_{\Sigma}}-$ coclosed forms, then for small $\tau>0$, the trajectory $\eta^\tau$ is transverse to the space of $d^{\star_{\Sigma^\tau}}-$coclosed forms. Therefore we could choose $\eta^\tau$ in such a way that $d^{\star_{\Sigma^\tau}}\eta^\tau=0$. In this case (\ref{eqn:ODE1}) yields $\gf_\tau$ constant.

Since $i^*_{\Sigma^\tau}(\star d\gf)=0$, then $\partial_\tau\gf_j=\partial_j\gf_\tau$. Thus $\star d\gf=\xi\wedge d\tau$ where $\xi$ has no normal components, $i^*_{\Sigma^\tau}(\star\xi)=0$, and $i^*_{\Sigma^\tau}(\xi)=i^*_{\Sigma^\tau}(\star d\gf)$. Thus as we differentiate $\star d\gf$, we take derivatives $\partial_i$ along $\Sigma^\tau$ and do not use derivatives $\partial_\tau$.  Therefore 
\[
	d\star d\gf=d\xi\wedge d\tau= 
	d\left((X_\Sigma^{\tau}\circ i^*_{\Sigma^\tau})(\star d\gf)\right)\wedge d\tau=0.
\]

Thus $\gf $ is a local solution for Yang-Mills such that $d^\star\gf=0$,  $\gf^N=0$ and $\gf^D=\phi^D$. This proves $	\ker \Lambda_{\tilde{\Sigma}_\geps}\simeq\ker d^{\star_\Sigma}. $
\end{proof}

Because of the uniqueness of the solution of the BVP (\ref{eqn:BVP-Sigma}), up to $\lambda \in\mathfrak{H}_D^1\left(\Sigma_\geps\right)$ the following statement follows.

\begin{tma}
The D-N operator $\Lambda_{\tilde{\Sigma}_\geps}$ does not depend on $\geps>0$. Thus we can define the D-N operator depending exclusively on the hypersurface $\Sigma$, as $\Lambda_{\Sigma}:=\Lambda_{\tilde{\Sigma}_\geps}$.
\end{tma}

Finally define a complex structure $J_{\Sigma}$ in	$\mathsf{L}_{\Sigma}:=(\ker d^{\star_{\Sigma}})^{\oplus 2}$ as in (\ref{eqn:J}).

%%%%%%%%%%%%%%%%%%%%%%%

\subsection{Complex structure for Euler-Lagrange solutions}

%%%%%%%%%%%%%%%%%%%%%%%

Define the D-N operator associated to the boundary $\partial M$ by considering the direct sum of the operators defined in (\ref{eqn:Lambda-Sigma}),
\[
	{\Lambda}_{\partial M}\left(\phi^D\right):=\bigoplus_{i=1}^m\Lambda_{\Sigma^i}\left(\phi^D_i\right)\in\bigoplus_{i=1}^m\Omega^k(\Sigma^i).
\]
Then there exists a complex structure $J_{\partial M}$ in
\[
	\ker \Lambda_{\partial M}\oplus\ran \Lambda_{\partial M}\simeq \mathsf{L}_{\partial M}.
\]

The complex structure $J_{\partial M}:\mathsf{L}_{\partial M}\rightarrow \mathsf{L}_{\partial M}$, defines also a complex structure in the affine space $\mathsf{A}_{\partial M}$, which is covariant with respect to translations.

\begin{tma}

The complex space $(\mathbf{L}_{M,\partial M},J)$  is a $J_{\partial M}-$complex subspace of $(\mathsf{L}_{\partial M},J_{\partial M})$.
\end{tma}

\begin{proof} The proof follows from the commutativity of the following diagram
\[\xymatrix{
	\ker\Lambda_{\partial M}\oplus \ran\Lambda_{\partial M}		\ar[r]^{J_{\partial M}}		&
	\ker\Lambda_{\partial M}\oplus \ran\Lambda_{\partial M}					\\
	\ker\Lambda_{ \tilde{M}}\oplus \ran\Lambda_{\tilde {M}}	\ar@{^{(}->}[u]			\ar[r]^{J}	&
	\ker\Lambda_{ \tilde{M}}\oplus \ran\Lambda_{\tilde {M}}	\ar@{^{(}->}[u]
}\]
Notice that for $1-$forms $(-1)^{(n-1)(k+1)}=1$, in the definition of $J_{\partial M}$ given in (\ref{eqn:J}) that uses (\ref{eqn:jN^-1}). 
\end{proof}

\begin{pro}
For the affine symplectic subspace $\mathbf{A}_{\partial M}:=\gf_0+\mathbf{L}_{M,\partial M}\subset \mathsf{A}_{\partial M}$, there exists a decomposition as $\mathbf{A}_{\tilde{M}}\oplus J_{\partial M}\mathbf{L}_{\tilde{M}}$ where $\mathbf{A}_{\tilde{M}}=\gf_0+\mathbf{L}_{\tilde{M}}$.
\end{pro}

\section{Hermitian structure}\label{sec:5}

%%%%%%%%%%%%%%%%%%%%%%%

In order to implement geometric quantization for the reduced space $\mathsf{A}_{\partial M}$, we need to describe a suitable Hermitian structure on $\mathsf{L}_{\partial M}$.

%%%%%%%%%%%%%%%%%%%%%%%

%%%%%%%%%%%%%%%%%%%%%%%

Let us consider a hypersurface $\Sigma$. The linear space ${\mathsf{L}_{\Sigma}}$ is completed to a complex separable Hilbert space with Hermitian metric, $\{\cdot,\cdot\}_{\Sigma},$ such that the imaginary part equals $\omega_{\Sigma}(\cdot,\cdot)=\frac{1}{2}\Im \{\cdot,\cdot\}_{\Sigma}$, while the real part is $ g_{\Sigma}(\cdot,\cdot):=\Re\{\cdot,\cdot\}_{\Sigma}$. Multiplication by $\sqrt{-1}$ in the complex Hilbert space structure can be defined in terms of a complex structure $J_{\Sigma}$ on the real vector space ${\mathsf{L}_{\Sigma}}$. The complex structure $J_{\Sigma}$ is tame with respect to the symplectic structure $\omega_{\Sigma}$,
\[
	g_{\Sigma}(\cdot,\cdot)=2\,\omega_{\Sigma}(\cdot,J_{\Sigma}\cdot).
\]
The positive definite bilinear form $g_{\Sigma}$ can be explicitly calculated as
\begin{equation}\label{eqn:g_partial M}
	g_{\Sigma}\left({\phi_1},{\phi_2}\right)=
		\int_{\Sigma}\phi_1^D\wedge\star_{\Sigma}\phi_2^D+ \phi_1^N\wedge\star_{\Sigma}\phi_2^N
		\qquad \forall {\phi_i}\in \mathsf{L}_{\Sigma},\, i=1,2.
\end{equation}
Define a tame complex structure $J_{\Sigma}:\mathsf{L}_{\Sigma}\rightarrow \mathsf{L}_{\Sigma}$ as in (\ref{eqn:J}),
\begin{equation}\label{eqn:Jexplicita}
	J_{\Sigma}\left(\phi^D,\phi^N\right)=
		\left(-\phi^N,\phi^D\right),
			\qquad \forall \phi=\left(\phi^D,\phi^N\right)\in \mathsf{L}_{\Sigma}.
\end{equation}

%%%%%%%%%%%%%%%%%%%%%%%%%%%%%%%%%%%%%%%%%%%%%%%%%%%%%%%%%%%

Involution under the complex linear product is conjugate, i.e. $\{\phi,\phi'\}_{\overline{\Sigma}}=\overline{\{\phi,\phi'\}_{\Sigma}}$ for all $\phi,\phi'\in L_\Sigma$. Thus for the gauge reduced spaces we hope we can fulfill the axioms for classical affine field theories as is shown in \cite{O,O1}.

%%%%%%%%%%%%%%%%%%%%%%%%%%%%%%%%%%%%%%%%%%%%%%%%%%%%%%

%%%%%%%%%%%%%%%%%%%%%%%

\section{Outlook: Holomorphic Quantization}

We have exposed the gauge symplectic reduction for abelian Yang-Mills fields. The aim of this work is to give a step towards geometric quantization of abelian Yang-Mills theories. We have completed the description for an affine field theory for gauge fields, see \cite{K,Segal:lectures} for scalar fields. The framework we have followed is the General Boundary Field Theory setting, see \cite{O3}. We have established the remaining main ingredients necessary for applying the tools exposed in \cite{O1,O}. Namely the existence of a complex structure, $J_{\partial M}$, taming the symplectic structure $\omega_{\partial M}$ in the space of boundary conditions modulo gauge. Another direction is the case of space-time regions with corners. Here the lack of differentiability of the boundary $\partial M$ on the stratified space of corners imposes difficulties in defining the complex structure.

\section{Acknowledgments} The author thanks Robert Oeckl for stimulating discussions. The author also thanks Centro de Ciencias Matem\'aticas-UNAM for the facilities and hospitality during the preparation of this work. This work was partially funded by Tecnol\'ogico de Monterrey Campus Morelia, M\'exico, and CONACYT-SEP M\'exico.

\bibliography{gaugequantum-bibtex}{}

\def\cprime{$'$}
\begin{thebibliography}{10}

\bibitem{T}
M.~E. Taylor, {\em Partial differential equations. {I}}, vol.~115 of {\em
  Applied Mathematical Sciences}.
\newblock Springer-Verlag, New York, 1996.
\newblock Basic theory.

\bibitem{BS}
M.~Belishev and V.~Sharafutdinov, ``Dirichlet to {N}eumann operator on
  differential forms,'' {\em Bull. Sci. Math.}, vol.~132, no.~2, pp.~128--145,
  2008.

\bibitem{O1}
R.~Oeckl, ``Affine holomorphic quantization,'' {\em J. Geom. Phys.}, vol.~62,
  no.~6, pp.~1373--1396, 2012.

\bibitem{O}
R.~Oeckl, ``Holomorphic quantization of linear field theory in the general
  boundary formulation,'' {\em SIGMA Symmetry Integrability Geom. Methods
  Appl.}, vol.~8, pp.~Paper 050, 31, 2012.

\bibitem{O2}
R.~Oeckl, ``Two-dimensional quantum {Y}ang-{M}ills theory with corners,'' {\em
  J. Phys. A}, vol.~41, no.~13, pp.~135401, 20, 2008.

\bibitem{CMR}
A.~S. Cattaneo, P.~Mnev, and N.~Reshetikhin, ``Classical and quantum
  {L}agrangian field theories with boundary,'' Proceedings of Science PoS:
  School and Workshop on Elementary Particle Physics and Gravity. Proceedings
  of the Corfu Summer Institute 044, Scuola Internazionale Superiore di Studi
  Avanzati (SISSA), 2011.

\bibitem{CMR1}
A.~S. Cattaneo, P.~Mnev, and N.~Reshetikhin, ``Semiclassical quantization of
  classical field theories,'' in {\em Mathematical Aspects of Quantum Field
  Theories} (D.~Calaque and T.~Strobl, eds.), Mathematical Physics Studies,
  ch.~Part III: (Semi-)Classical Field Theories, pp.~275--314, Springer-Verlag,
  2015.

\bibitem{Zu}
G.~J. Zuckerman, ``Action principles and global geometry,'' in {\em
  Mathematical aspects of string theory ({S}an {D}iego, {C}alif., 1986)},
  vol.~1 of {\em Adv. Ser. Math. Phys.}, pp.~259--284, World Sci. Publishing,
  Singapore, 1987.

\bibitem{O3}
R.~Oeckl, ``General boundary quantum field theory: foundations and probability
  interpretation,'' {\em Adv. Theor. Math. Phys.}, vol.~12, no.~2,
  pp.~319--352, 2008.

\bibitem{Seg:cftdef}
G.~Segal, ``The definition of conformal field theory,'' in {\em Topology,
  geometry and quantum field theory}, pp.~421--577, Cambridge: Cambridge
  University Press, 2004.

\bibitem{K}
S.~Kandel, ``Functorial quantum field theory in the riemannian setting,'' {\em
  Advances in Theoretical and Mathematical Physics}, vol.~20, no.~6,
  pp.~1443--1471, 2016.

\bibitem{Segal:lectures}
G.~Segal, ``Three roles of quantum field theory.''
  http://www.mpim-bonn.mpg.de/node/3372/program/, May 2011.

\bibitem{SS}
V.~Sharafutdinov and C.~Shonkwiler, ``The complete dirichlet-to-neumann map for
  differential forms,'' {\em Journal of Geometric Analysis}, vol.~23, no.~4,
  pp.~2063--2080, 2013.

\bibitem{Wo}
N.~M.~J. Woodhouse, {\em Geometric quantization}.
\newblock Oxford Mathematical Monographs, The Clarendon Press, Oxford
  University Press, New York, second~ed., 1992.
\newblock Oxford Science Publications.

\bibitem{At}
M.~Atiyah, ``Topological quantum field theories,'' {\em Inst. Hautes \'Etudes
  Sci. Publ. Math.}, no.~68, pp.~175--186 (1989), 1988.

\bibitem{We}
A.~Weinstein, ``Symplectic categories,'' {\em Port. Math.}, no.~67,
  pp.~261--278, 2010.

\bibitem{DM}
H.~G. D{\'{\i}}az-Mar{\'{\i}}n, ``General boundary formulation for
  {$n$}-dimensional classical abelian theory with corners,'' {\em SIGMA
  Symmetry Integrability Geom. Methods Appl.}, vol.~11, pp.~048, 35 pages,
  2015.

\bibitem{Sc}
G.~Schwarz, {\em Hodge decomposition---a method for solving boundary value
  problems}, vol.~1607 of {\em Lecture Notes in Mathematics}.
\newblock Springer-Verlag, Berlin, 1995.

\bibitem{CW}
{\v{C}}.~Crnkovi{\'c} and E.~Witten, ``Covariant description of canonical
  formalism in geometrical theories,'' in {\em Three hundred years of
  gravitation}, pp.~676--684, Cambridge Univ. Press, Cambridge, 1987.

\bibitem{Jost-RG}
J.~Jost, {\em Riemannian geometry and geometric analysis}.
\newblock Universitext, Springer-Verlag, Berlin, second~ed., 1998.

\bibitem{MariniDN}
A.~Marini, ``Dirichlet and neumann boundary value problems for yang-mills
  connections,'' {\em Communications on Pure and Applied Mathematics}, vol.~45,
  no.~8, pp.~1015--1050, 1992.

\bibitem{Michor}
P.~Michor, {\em Topics in Differential Geometry}.
\newblock Graduate Studies in Mathematics, USA: American Mathematical Society.

\end{thebibliography}
\bibliographystyle{ieeetr}
 
\end{document}